\documentclass[11pt]{article}
\usepackage[utf8]{inputenc}

\usepackage[margin=1in]{geometry}

\usepackage{graphicx}
\usepackage{complexity}
\usepackage{multirow}
\usepackage{amsmath,amssymb,amsfonts}
\usepackage{amsthm}
\usepackage{mathrsfs}
\usepackage{mathtools}
\usepackage{xspace}
\usepackage{booktabs}
\usepackage{tabularx}
\usepackage{multicol}
\usepackage{csquotes}
\usepackage{thm-restate}
\usepackage{cite}
\usepackage{tikz}
\usetikzlibrary{arrows,calc,arrows.meta,decorations.pathreplacing}
\usepackage[colorlinks=true,linkcolor=black, urlcolor=black, citecolor=black]{hyperref}
\usepackage[capitalise]{cleveref}
\hypersetup{
    linkcolor=rwthred,
    citecolor=rwthgreen,
    urlcolor=rwthblue,
    menucolor=rwthred
}
    
\newtheorem{theorem}{Theorem}[section]
\newtheorem{lemma}[theorem]{Lemma}
\newtheorem{proposition}[theorem]{Proposition}
\newtheorem{corollary}[theorem]{Corollary}

\newtheorem{observation}[theorem]{Observation}
\newtheorem{example}[theorem]{Example}

\definecolor{rwthblue}{cmyk}{1, 0.5, 0, 0}
\definecolor{rwthred}{cmyk}{0.15, 1, 1, 0}
\definecolor{rwthbordeaux}{cmyk}{.25, 1, .7, .2}
\definecolor{rwthgreen}{cmyk}{0.70, 0, 1, 0}

\tikzset{bluevertex/.style={circle, fill= rwthblue, inner sep=0pt,minimum size=6pt}}
\tikzset{redvertex/.style={circle, fill= rwthred, inner sep=0pt,minimum size=6pt}}

\newcommand{\intervalnominal}[5]{\draw[Bracket-Bracket,thick,rwthred] (#2, -.6*#1) -- node[above=-2pt,black]{#4} node[below=-1pt,black]{#5} (#3, -.6*#1);}

\newcommand{\intervaliscblue}[5]{\draw[line width = 1mm,rwthblue!20] (#2, -.6*#1) -- node[above=-2pt,black]{#4} node[below=-1pt,black]{#5} (#3,-.6*#1);}

\newcommand{\interval}[6][black]{\draw[Bracket-Bracket,thick,#1] (#3, -#2*.7-1.3) -- node[above=-2pt,black]{#5} node[below=-1pt]{#6} (#4, -#2*.7-1.3);}

\newenvironment{ichart}[3][1]{
	\begin{tikzpicture}[xscale=#1]
		\pgfmathsetmacro{\yb}{-.6*#2-.5}
		\draw[-latex'] ($(-.5, \yb)$) -- ($(#3, \yb) + (.5, 0)$);
		\foreach \t in {0, ..., #3}{
			\draw[-] ($(\t, \yb)$) -- ($(\t, \yb) + (0, -.1)$);
		}
	}{\end{tikzpicture}}

\newcommand{\tikzxmark}{%
	\tikz[scale=0.35] {
		\draw[thick,line width=0.7,line cap=round,rwthred] (0,0) to [bend left=6] (1,1);
		\draw[thick,line width=0.7,line cap=round,rwthred] (0.2,0.95) to [bend right=3] (0.8,0.05);
}}
\tikzset{
	small vertex/.style={draw,circle,minimum size=6pt,inner sep=0pt}}

\newcommand{\f}{f}
\newcommand{\e}{e}
\newcommand{\F}{F}
\newcommand{\Rem}{R}

\newcommand{\backup}{\textbf{b}}

\newcommand{\matching}{{\sc Bipartite Matching}\xspace}

\newcommand{\robustmatching}{{\sc Robust Bipartite Matching}\xspace}
\newcommand{\robustmatchingshort}{{\sc RBM}\xspace}

\newcommand{\tsat}{{\sc3-SAT}\xspace}

\newcommand{\robuststableset}{{\sc Robust Stable Set}\xspace}
\newcommand{\robuststablesetshort}{{\sc RSS}\xspace}

\newcommand{\robuststablesetinterval}{{\sc Robust Interval Stable Set}\xspace}

\newcommand{\stablesetbipartite}{{\sc Bipartite Stable Set}\xspace}

\newcommand{\robuststablesetbipartite}{{\sc Robust Bipartite Stable Set}\xspace}

\newcommand{\weightedind}{{\sc Robust Weighted Stable Set}\xspace}

\newcommand{\generalproblem}{{Recoverable Robust Optimization with Commitment}\xspace}
\newcommand{\repairable}{{repairable}\xspace}

\newcommand{\private}{{private}\xspace}
\newcommand{\universal}{{universal}\xspace}
\newcommand{\extra}{{extra}\xspace}

\newcommand{\boundedregretintervalselection}{{\sc $\lambda$-Robust Interval Scheduling}\xspace}
\newcommand{\boundedregretintervalselectionshort}{{\sc $\lambda$-RIS}\xspace}

\newcommand{\issh}{{\sc Interval Scheduling with Colors}\xspace}
\newcommand{\Issh}{{\sc ISC}\xspace}

\newcommand{\robustmatroidbasis}{{\sc Robust Matroid}\xspace}
\newcommand{\robustmatroidbasisshort}{{\sc RM}\xspace}

\newcommand{\repairablematching}{{\sc Repairable Bipartite Matching}\xspace}
\newcommand{\repairablematchingshort}{{\sc RepBM}\xspace}
\newcommand{\irm}{$I_\text{\repairablematchingshort}$\xspace}

\newcommand{\intervalselection}{{\sc Interval Scheduling}\xspace}
\newcommand{\intervalselectionshort}{{\sc InS}\xspace}
\newcommand{\robustintervalselection}{{\sc Robust Interval Scheduling}\xspace}
\newcommand{\robustintervalselectionshort}{{\sc RIS}\xspace}
\newcommand{\problempi}{\ensuremath{\Pi}\xspace}

\newcommand{\X}[1]{}

\DeclareMathOperator*{\argmax}{arg\,max}

\def\phi{\varphi}

\makeatletter
\newcommand{\problemtitle}[1]{\def\PROBLEMTITLE{#1}}% Store problem title
\newcommand{\probleminput}[1]{\def\PROBLEMINPUT{#1}}% Store problem input
\newcommand{\problemquestion}[1]{\def\PROBLEMQUESTION{#1}}% Store problem question
\newenvironment{problemenv}{
\problemtitle{}\probleminput{}\problemquestion{}% Default input is empty
}{
\par\addvspace{.5\baselineskip}
\noindent\begin{tabularx}{\linewidth}{|@{\hspace{5pt}} l@{\hspace{5pt}} X l}
	\multicolumn{2}{@{}l}{\sc\PROBLEMTITLE} \\% Title
	\textbf{Given:} & \PROBLEMINPUT \\% Input
	\textbf{Task:} & \PROBLEMQUESTION% Question
\end{tabularx}
\par\addvspace{.5\baselineskip}
}
\makeatother

\title{Recoverable Robust Optimization with Commitment\footnote{This article has been published in Mathematical Programming~\cite{hommelsheim2023recoverablejournal}. Moreover, this work was conducted when the first and the third authors were affiliated to the University of Bremen and RWTH Aachen University, respectively. }}
\author{
	Felix Hommelsheim\footnote{University of Cologne, Cologne, Germany. \texttt{hommelsheim@cs.uni-koeln.de}}
	\and Nicole Megow\footnote{University of Bremen, Germany. \texttt{nicole.megow@uni-bremen.de}}
	\and Komal Muluk\footnote{TU Dortmund University, Germany. \texttt{komal.muluk@math.tu-dortmund.de}}
	\and Britta Peis\footnote{RWTH Aachen University, Germany. \texttt{peis@oms.rwth-aachen.de}}
}
\date{\today}

\begin{document}
	\maketitle
		
	\begin{abstract}
		We propose a model for recoverable robust optimization with {\em commitment}. 
		Given a combinatorial optimization problem and uncertainty about elements that may fail, we ask for a robust solution that, after the failing elements are revealed, can be augmented in a limited way. 
		However, we commit to preserve the 
		non-failing elements of the initial solution. 
		We settle the computational complexity of such a robust counterpart of various classical polynomial-time solvable combinatorial optimization problems.
		We show, for the weighted matroid independent set problem, that an optimal solution to the nominal problem is also optimal for its robust counterpart.
		Indeed, matroids are provably the only structures with this strong property. 
		Robust counterparts of other problems are \NP-hard such as the matching problem and the stable set problem, even in bipartite graphs. 
		However, we establish polynomial-time algorithms for the robust counterparts of the unweighted stable set problem in bipartite graphs and the weighted stable set problem in interval graphs, also known as the interval scheduling problem.
	\end{abstract}
	
	\section{Introduction}\label{sec:Intro}

Robust optimization has become a central framework in operations research for addressing uncertainty in decision-making and optimization. Unlike traditional optimization, where parameters are assumed to be known precisely, robust optimization considers \emph{scenarios} where some parameters are uncertain but lie within a predefined uncertainty set, which may be specified implicitly. 
Early research in this area focused on computing static solutions that optimize the worst-case cost. For comprehensive overviews, we refer to~\cite{KouvelisYu1997-book,Ben-talGN2009-book-robust,buchheim2018robust}.
However, it is well recognized that this approach may lead to overly conservative solutions as a single, potentially very unlikely, scenario might heavily influence the outcome.

To address this issue of conservatism in robust optimization and introduce more flexibility, 
various frameworks have been proposed, including two-stage adjustable robust optimization by~\cite{ben2004adjustable}, recoverable robust optimization by~\cite{liebchen2009concept}, and demand robust optimization by~\cite{dhamdhere2005pay}. 
These frameworks share a common two-stage structure: In the first stage, a partial solution is determined based on the uncertainty set, while in the second stage, after the uncertain scenario has been revealed, the initial solution can be adjusted within some predefined recovery bounds. This second-stage adjustment is referred to as \emph{recourse}. 

In this work, we use the notion of {\em recoverable robust optimization} as introduced by Liebchen, L{\"u}bbecke, M{\"o}hring and Stiller~\cite{liebchen2009concept}. 
Initially developed in the context of timetabling, linear programming, and railway optimization, it seeks first-stage solutions that can be adapted feasibly across all scenarios within specified recovery limits; the corresponding recovery limits are modeled, for example, via the size of the symmetric difference of solutions.
The cost is typically evaluated as the sum of the first- and second-stage costs, potentially under different cost functions.
This model has been successfully applied to a variety of combinatorial optimization problems such as the shortest path problem \cite{busing2012recoverable,DBLP:journals/corr/abs-2403-20000,JackiewiczKZ2024-networks}, spanning tree and more general matroid problems \cite{hradovich2017recoverable,DBLP:journals/ol/HradovichKZ17,lendl2021matroid}, perfect matchings \cite{dourado_et_al_15}, scheduling problems \cite{DBLP:journals/dam/BoldG22}, 
knapsack problems \cite{DBLP:journals/ejco/BusingGKK19,DBLP:conf/inoc/BusingKK11}, and selection problems \cite{DBLP:journals/dam/LachmannLW21}
which led to insights into the computational complexity, as well as the development of exact and approximate algorithms for solving these problems.

A major drawback of recourse models, including the recoverable robustness framework, is that they do not necessarily preserve the elements of the initial solution during recourse.
However, such preservation is crucial in many practical applications, particularly when commitments or promises have been made.
In any reservation-based system---such as hotel bookings, rental car reservations, or airline seat reservations---the supplier guarantees services to the customers holding reservations. 
If a customer cancels her reservation, the supplier may modify 
the set of reserved customers, but must ensure that the previously reserved customers who did not cancel remain part of the modified solution.
In such scenarios, the supplier can only {\em augment} the existing set of reserved (and not canceled) customers with new customers, without altering the existing commitments.

This motivates us to introduce our new model \emph{\generalproblem}, where we commit to the first-stage decision, 
and restrict the recourse action to augmenting the partial solution which remains after the cancellation or after some other uncertain disruption has taken place. More formally, we define the problem as follows.

\paragraph{\generalproblem.} Consider a combinatorial optimization problem \problempi with ground set $E$, a set $\mathcal{F} \subseteq 2^E$ of feasible solutions, and a weight function $w \colon E \rightarrow \mathbb{R}_+$, where the task is to find a solution $S \in \mathcal{F}$ maximizing $w(S) = \sum_{e \in S} w(e)$.
In the first stage in our model, 
we choose a solution $S \in \mathcal{F}$ to the underlying problem \problempi.
Then, in the second stage, a set $\F \subseteq E$ of \emph{at most} $k$ elements is 
deleted, followed by a step in which a set $\Rem \subseteq E\setminus \F$ of size at most $\ell$ elements is
added to $S \setminus \F$ such that the resulting set is still a solution, i.e., $S' = (S \setminus \F) \cup \Rem \in \mathcal{F}$.
We call $S$ the \emph{first-stage} solution and $S'$ the \emph{second-stage} solution.
The goal is to select a first-stage solution $S$ that maximizes the weight of the second-stage solution, $w(S')$ in the worst case.

We refer to the set $\F$ that is deleted as the interdiction and the set $\Rem$ that is added as \emph{recourse}.
The interdiction can be seen as an \emph{adversary}, who aims to delete a set $\F \subseteq E$ that decreases the final objective value $w(S')$ as much as possible. The adversary can be assumed to have full knowledge 
of the set $\Rem \subseteq E \setminus \F$ that we would optimally add during the recourse phase following the deletion of~$\F$.
Thus, we solve the following problem:
\begin{align}\label{eq: (k,l)recoverablerobustproblem}
	\max_{S\in \mathcal{F}}~~ \min_{\substack{\F \subseteq E \\ |\F| \leq k }}~~ \max_{\substack{\Rem \subseteq E \setminus \F, ~ |\Rem| \leq \ell\\
			(S \setminus \F) \cup \Rem \in \mathcal{F}}}~~  w((S \setminus \F) \cup \Rem)\ . \tag{$k, \ell$-RP}
\end{align}

For the sake of simplicity, we assume that $E \leftarrow E \cup \{\emptyset\}$, i.e., $\{ \emptyset \} \in E$, and we define $w(\{\emptyset\})\coloneqq 0$. 
The inclusion of the empty set in $E$ allows us to conveniently represent scenarios involving empty deletions or additions
of the elements in~$E$.
Further, throughout the paper, we consider downward-closed set systems $\mathcal{F}$, meaning that if $X' \subseteq X$ and $X \in \mathcal{F}$, then $X' \in \mathcal{F}$.
Additionally, for a set $X$ and an element $g$, we use $X+g \coloneqq X\cup\{g\}$ and $X-g \coloneqq X\setminus{\{g\}}$.

A significant part of this paper will deal with the special case of the (\ref{eq: (k,l)recoverablerobustproblem}) 
with $k=1$ and $\ell = 1$. 
Formally we define the restricted version where $k=1$ and $\ell=1$ as follows:
\begin{align}
	\label{eq: (1,1)recoverablerobustproblem}
	\max_{S\in \mathcal{F}}~~ \min_{\f\in E}~~ \max_{\substack{\e\in E-\f \\ S-\f+\e \in \mathcal{F}}}~~ w(S-\f+\e) \ . \tag{RP}
\end{align}

Without loss of generality, we assume that the interdiction is non-empty. 

For a combinatorial optimization problem \problempi, we call the robust variant of~\problempi as in (\ref{eq: (1,1)recoverablerobustproblem}) the \emph{robust counterpart of \problempi}, or {\sc Robust} \problempi in short.
We usually call the problem \problempi, i.e., $\max \{w(f)\mid f\in \mathcal{F}\}$, the \emph{underlying} or \emph{nominal} problem.

As an example, consider a car-rental company with a single car, which we wish to rent out to different customers to maximize the total revenue. 
All customers have their own time intervals in which they want to rent the car exclusively and a price they are willing to pay.
If there is no uncertainty,
this is equivalent to an instance of the \emph{interval scheduling problem}, where we are given a set of weighted intervals on the real line, and the goal is to find a maximum-weighted set of pairwise-disjoint intervals.
However, after the selection of the initial solution, some customers might withdraw from the 
contract, leading to less revenue for our company.
In this case, we could still add unserved customers to our remaining initial solution.
Although the additional customers should not be in conflict with our remaining initial solution, as those reservations were firm.
As a concrete example, consider the five intervals $(1, 3)$, $(2, 5)$, $(4, 7)$, $(6, 9)$, $(8, 10)$ with weights $10, 8, 2, 8, 10$, respectively, each representing a customer.
Picking $(1, 3)$, $(4, 7)$, and $(8, 10)$ is an optimal solution to the underlying problem with value $22$.
An optimal solution to the robust counterpart, however, is to pick only $(1, 3)$ and $(8, 10)$, knowing that if $(1, 3)$ or $(8, 10)$ is interdicted, in the recourse one can add $(2, 5)$ or $(6, 9)$, respectively. 
This leads to a worst-case weight of 18 after the interdiction and recourse.
On the other hand, when considered the optimal solution to the underlying problem, i.e., picking $(1, 3)$, $(4, 7)$, and $(8, 10)$, the deletion of interval $(1, 3)$ would lead to a final value of 12 in the robust version, as the interval $(4, 7)$ blocks the recourse.
This shows that selecting an optimal solution to the nominal problem need not be optimal for its robust counterpart.

\paragraph{The Adversary Problem.} A closely related problem is the \emph{adversary problem} of~\eqref{eq: (k,l)recoverablerobustproblem}. It 
is the subproblem of~\eqref{eq: (k,l)recoverablerobustproblem} in which the set $S \in \mathcal{F}$ is already fixed and the adversary needs to find the optimal interdiction set $\F$.
That is, given some $S \in \mathcal{F}$, the task is to solve the following problem:
\begin{align}
	\label{eq:(k,l)adversaryproblem}
	\min_{\substack{\F \subseteq E \\ |\F| \leq k }}~~ \max_{\substack{\Rem \subseteq E \setminus \F, ~ |\Rem| \leq \ell\\
			(S \setminus \F) \cup \Rem \in \mathcal{F}}}~~  w((S \setminus \F) \cup \Rem)\ . \tag{$k, \ell$-AP}
\end{align}

Note that if $\ell \geq |E|$ for~\eqref{eq: (k,l)recoverablerobustproblem}, it is optimal to choose an empty first-stage solution, as the recourse is essentially unbounded.
Hence, for $\ell \geq |E|$, we can assume $S = \emptyset$. For this case, problem~\eqref{eq:(k,l)adversaryproblem} is often called the \emph{interdiction version} of the problem and is also known as the \emph{most vital edge problem}.
For many underlying problems, its interdiction version is \NP-hard if $k$ is part of the input, including spanning tree interdiction \cite{frederickson1999increasing,zenklusen20151}, matching interdiction \cite{zenklusen2009blockers,zenklusen2010matching}, or shortest path interdiction \cite{bar1998complexity}. 
However, for constant $k$ the adversary problem
is clearly polynomial time solvable by simple enumeration.

\paragraph{Bounded-Regret Problem.}
Another closely related problem, which is both independently interesting and crucial for parts of our results, is the \emph{bounded-regret} version of our problem.  	
In the setting of our model, after the interdiction of an element $f\in E$, one takes the recourse action that may add an element 
$e \in E - f$
to the set $S - f$ which maximizes the total weight of the final solution.
However, the weight of the final solution $w(S-f+e)$ could be less than the weight of the initial solution $w(S)$. 
This weight loss, $w(S)-w(S-f+e)$, can be seen as a \emph{regret}; we denote it by $\Delta(f,S)$ and it is formally defined as:
\[
\Delta(f, S) = w(S) - \max \{w(S-f+e) \mid e\in E-f, S-f+e \in \mathcal{F}\}\ .
\]

The goal in the bounded-regret problem is to find a maximum-weight solution $S$ such that, for each  element $f\in E$, the interdiction of $f$ followed by the construction of a maximum-weight second-stage solution $S-f+e$,
results in a regret bounded by some given value $\lambda \in \mathbb{R}$; formally:
\begin{align}
	\label{eq:opt:sol:bounded-regret}
	\max_{\substack{S\in \mathcal{F}\\ \Delta(f, S) \leq \lambda: \forall f\in E }} w(S) \ . \tag{$\lambda$-RP}
\end{align}

We note that for some values of $\lambda$ the problem~\eqref{eq:opt:sol:bounded-regret} can be infeasible. We denote by $w_{\text{opt}}(\lambda)$ the objective value of \eqref{eq:opt:sol:bounded-regret} for an optimal solution, if one exists.
The bounded-regret version for general values of $k$ and $\ell$ of the problem is defined analogously. 
However, in this work, we focus on the above mentioned case with $k = \ell = 1$.
Now, notice the tight connection between the bounded-regret version \eqref{eq:opt:sol:bounded-regret} and the robust counterpart \eqref{eq: (1,1)recoverablerobustproblem} of a problem.
If we have a polynomial-time algorithm for \eqref{eq:opt:sol:bounded-regret}, and if there are only polynomial-many different candidates for $\lambda$ values, then one can enumerate over all $\lambda$ values and solve \eqref{eq:opt:sol:bounded-regret} for each one of them.  
Finally, to solve \eqref{eq: (1,1)recoverablerobustproblem}, we are interested in the $\lambda$ value which maximizes  $w_{\text{opt}}(\lambda) - \lambda$. 
The idea is that, if we choose a set \( S \subseteq E \) corresponding to the solution \( w_{\text{opt}}(\lambda) \) while selecting the first-stage solution of \eqref{eq: (1,1)recoverablerobustproblem}, the maximum regret, even after adversarial interdiction and recourse, will be at most \( \lambda \).
Thus, the corresponding value $\max_{\lambda} w_{\text{opt}}(\lambda) - \lambda$ gives us an optimal solution of \eqref{eq: (1,1)recoverablerobustproblem}.   
Now, observe that for $k = \ell = 1$ the number of possible $\lambda$ values is at most $O(|E|^2)$ and hence polynomially~bounded.
Note that $\lambda$ can also be negative and hence we also enumerate over these~choices.

\subsection{Our Contribution}\label{sec:contribution}
\begin{table}
	\centering
	\begin{tabular}{l|cc}
		\toprule
		& unweighted & weighted \\
		\midrule
		$(k, \ell)$-\robustmatroidbasis \ for $k \leq \ell$
		& Polynomial & Polynomial  \\ 
		$(1, p)$-\robuststablesetinterval \ for constant $p$ & Polynomial  & Polynomial \\ 
		$(1, 1)$-\robustmatching 
		& \NP-hard  & \NP-hard  \\
		$(1, 1)$-\robuststablesetbipartite 
		& Polynomial  &  \NP-hard \\
		\bottomrule
	\end{tabular}
	\caption{Our results. Note that hardness results for the $(1, 1)$-case also imply hardness for the general $(k, \ell)$-case.}
	\label{tab:results}
\end{table}
We present algorithms and complexity results for a variety of combinatorial optimization problems in the model of \generalproblem. 
In this paper, we focus on some underlying problems that are classical downward-closed combinatorial optimization problems which can be solved in polynomial time.
As one would expect, robust counterparts of some of these
problems turn out to be \NP-hard, even in the setting $k = \ell = 1$.
However, in some cases the robust counterpart indeed remains polynomial-time solvable.
The most prominent problem in this category is the \robustmatroidbasis\ problem. 
We show that \robustmatroidbasis is polynomial time solvable if $k \leq \ell $.
In fact, we show that it suffices to solve the nominal problem for obtaining an optimal first-stage solution for the robust counterpart.
Surprisingly, this does not hold as soon as $k>\ell$.
For other problems, solving the robust counterpart is much more involved. 
Table~\ref{tab:results} summarizes our results. 

A major challenge in designing algorithms for the robust problem lies in the fact that a first-stage solution $S$ must be robust against any interdiction with possible recourse taken into account, even if $k=\ell=1$. 
Intuitively, we may pre-compute for each element $f\in E$, a \enquote{backup element} for the scenario in which  the element $f$ fails.
We then solve \enquote{the problem with backups}. 
However, this increases the problem complexity substantially as such backups must not be selected for the solution $S$ and the uncertainty about the failing element (and its activated backup) causes complex dependencies.

Next, we discuss our main results for the variety of problems; definitions follow later.

\paragraph{Robust Matroid (RM).}
We study the problem of finding a maximum-weight independent set of a matroid in the robust setting.
We refer to this problem as \robustmatroidbasis(\robustmatroidbasisshort for short).
We show that an optimal solution to the nominal problem, i.e., the problem of finding a maximum-weight independent set in a matroid, is also an optimal first-stage solution for its robust counterpart, whenever $k \leq \ell$.

Hence, the classical greedy algorithm for the nominal problem implies a polynomial-time algorithm for \robustmatroidbasisshort.
This is a bit surprising since the adversary problem 
of a special case of the nominal problem, spanning tree interdiction, is \NP-hard in general.
However, this is not a contradiction since in our algorithm and proof we do not need to actually compute the adversary's response and can simply prove that computing an optimal solution to the nominal problem is also optimal for~\robustmatroidbasisshort.

Further, we give two cases in which this property is not true:
First, we show that for $k > \ell$, the optimum solution to the nominal problem need not be an optimum solution to \robustmatroidbasisshort, even if $k=2$ and $\ell =1$.
Second, we also show that for any non-matroid this property is not true, even if $k=\ell=1$. 
Therefore, for $k \leq \ell$, matroids are exactly those structures for which the \emph{price of robustness}, i.e., the worst-case ratio between the weight of an optimal solution to the underlying problem and the weight of an optimal first-stage solution to its robust counterpart, is guaranteed to be one.

\paragraph{Robust Stable Set in Interval Graphs (RIS).}
We give a polynomial-time algorithm for $(1, 1)$-\robuststablesetinterval, the stable set problem in interval graphs. 
For convenience, we view this problem as selecting pairwise-disjoint intervals from a given set of intervals and call it \robustintervalselection (\robustintervalselectionshort).
Our key ingredient is a reduction of \robustintervalselectionshort to its bounded-regret version.
In general, it is not straightforward to construct a polynomial-time algorithm for the bounded-regret problem, even in the case of the matroid independent set 
problem for which we know that the robust counterpart can be solved efficiently.
As a crucial step, we show that for a given $\lambda$, the bounded-regret \robustintervalselection problem can be reduced to a more general interval scheduling 
problem, which we call \issh (\Issh) and for which we give an efficient dynamic program. 
Furthermore, we show how this result can be generalized to $(1, p)$-\robuststablesetinterval for any constant $p \in \mathbb{Z}_{\geq 1}$.

\paragraph{Robust Bipartite Matching (RBM).} 
Even for the unweighted case, we show that \robustmatchingshort is \NP-hard; we give a reduction from \tsat.
Here, the underlying problem \matching can be viewed as a special case of a matroid intersection problem.
Hence, our complexity result for \robustmatchingshort reveals a drastic increase in problem difficulty when the underlying problem involves finding a maximum-weight independent set 
of the {\em intersection of two matroids} instead of a {\em single matroid}.
This is somewhat surprising, as for a single matroid, even an optimal solution of the nominal problem is optimal for its robust counterpart.

\paragraph{Robust Bipartite Stable Set (RBSS).}
A maximum-weight stable set in a graph is the complement of a minimum-weight vertex-cover, which in bipartite graphs can be found (in polynomial-time) via a primal-dual algorithm which simultaneously computes both, a maximum-weight matching, and a minimum-weight vertex cover.
Hence, the complexity of \stablesetbipartite and \matching is the same and even the algorithms are very similar (which inspired us to look into these problems more closely).
However, the complexity of the robust counterpart of bipartite stable set, RBSS, depends highly on the weights.
We show that, for unit weights, RBSS is polynomial-time solvable (even in more general \emph{König-Egerváry} graphs).
We prove this by reducing the problem to finding \repairable maximum stable sets, which are maximum stable sets $S$ satisfying that for each $v \in S$ there is some $u \in V \setminus S$ such that $S - v + u$ is a stable set.
We then show that \repairable maximum stable sets can be found in polynomial time (if they exist).
Further, we show that RBSS is \NP-hard in the weighted setting.

\subsection{Further Related Work}\label{sec:related-work} 

Recoverable robust optimization is a powerful framework for decision-making under uncertainty that combines robustness with limited adaptability. 
Due to its generality, many (robust) combinatorial optimization problems can be naturally interpreted as recoverable robust problems.
Most research in this area has focused on settings where the second-stage {\em costs} are uncertain, and  recourse actions are constrained by a parameter that restricts the {\em symmetric difference} between the first-stage and second-stage solution. 
When facing discrete uncertainty---where a finite list of potential cost functions is explicitly given---many recoverable robust problems 
turn out to be \NP-hard, even for problems that are polynomial-time solvable in the full-information setting, such as the shortest path problem \cite{busing2012recoverable}, 
basic scheduling and selection problems \cite{KouvelisYu1997-book,DBLP:journals/ol/BusingKK11,DBLP:journals/eor/GoerigkLW22a}.
To address this issue, approximation algorithms have been proposed; see, e.g.,~\cite{DBLP:journals/ol/ChasseinG16,DBLP:journals/eor/GoerigkLW22a}. 

In contrast, under interval uncertainty, these problems often 
reduce  to finding two solutions to the nominal problem minimizing the weight for two different cost functions, while ensuring that the size of the symmetric difference is bounded.
This reduction has been applied to problems such as the selection problem \cite{DBLP:journals/dam/LachmannLW21}, spanning trees \cite{hradovich2017recoverable}, general matroids \cite{lendl2021matroid}, the traveling salesperson problem \cite{goerigk2021recoverable}, and the assignment problem \cite{DBLP:conf/iwpec/0001HLW21}. 

While most prior work focuses on cost uncertainty, some recent works have begun assuming {\em structural} uncertainty \cite{dourado_et_al_15,ito2022parameterized}, as we do in this paper. 
For instance, \cite{dourado_et_al_15} studied the following related model:  
given a bipartite graph and a set of uncertain edges $F$, the task is to find a perfect matching $M$ such that for any subset $F' \subseteq F$ with $|F'|\leq r$, the graph $G-F'$ contains a perfect matching $M'$ where the symmetric difference of $M$ and $M'$ is bounded by some parameter.
\cite{dourado_et_al_15} establish hardness results for this problem and give characterizations for graphs that admit such matchings $M$.

Similar two-stage concepts that introduce a degree of adaptivity in robust optimization have been explored
under different names. One example is {\em adjustable robustness}, studied by 
\cite{ben2004adjustable} in the context of robust linear programming and applications to multi-stage inventory management problem. Here, certain solution variables have to be fixed in the first stage, before the uncertain data is revealed (``here and now'' decisions), whereas other variables can be adjusted in a second stage, after the uncertain data becomes known (``wait and see'' decisions). 

In contrast, in {\em demand robustness} or {\em two-stage adjustable robustness}, all solution variables are free to be adjusted in the second stage. However, the extent of these adjustments is regulated via a cost function, where the cost for setting a variable in the second stage (typically selecting it to be part of the solution) is substantially higher than its first-stage cost. 
This model has been studied extensively from both complexity and approximation perspectives for problems such as certain linear programs \cite{DBLP:journals/mor/BertsimasG10,DBLP:journals/ior/ZhenHS18,DBLP:journals/mp/HousniFG24}, a variety of discrete covering problems including 
shortest paths, Steiner trees, 
set cover, and facility location \cite{dhamdhere2005pay,feige2007robust,khandekar2013two,gupta2014thresholded,DBLP:conf/ipco/HousniGS21}, matching \cite{DBLP:conf/approx/HousniGH021}, and scheduling problems \cite{DBLP:conf/approx/ChenMRS15}.

In all these different models, a commitment to the first-stage solution is not explicitly required. Instead, recourse is constrained by restricting the symmetric difference between first-stage and second-stage cost and/or by applying different cost functions for the two stages. 
In this paper, we {\em explicitly require commitment} to the first-stage solution while allowing limited adaptation in the second stage. 
This adaptation is achieved by {\em adding} (selecting) a restricted  number of additional solution variables in the second stage. 
Notably, similar fine-grained restrictions on recourse budgets have been explored in some problem-specific contexts. For example, B{\"u}sing, Koster, and Kutschka explicitly limit the number of items that can be added or removed from the first-stage solution for the knapsack problem, in an experimental context \cite{DBLP:journals/ol/BusingKK11}. 
Further, \cite{DBLP:journals/dam/BoldG22} require that a certain number of jobs must have the same scheduling position in the first and second stage in a single-machine scheduling problem. With our model, we aim to provide a general framework for recoverable robustness with commitment, unifying and extending such problem-specific approaches.

In this paper, we focus on downward-closed problems, where a solution remains feasible even if arbitrary elements are removed. Hence, the structural uncertainty impacts only the solution quality (cost) and not its feasibility. This proves to be both interesting and challenging within our model.

A somewhat opposing setting is captured by {\em bulk-robustness} and {\em flexible network design}. In this framework, 
we select a superset of a solution that remains feasible even if a subset of elements from a given list is deleted.
Hence, considering upward-closed problems, the goal is to find a minimum-cost solution that remains feasible in any scenario.
Bulk-robustness has been studied in various contexts such as connectivity and matchings \cite{adjiashvili2015bulk,adjiashvili2022flexible,hommelsheim2021secure,DBLP:conf/icalp/ChekuriJ23,DBLP:journals/mp/BoydCHI24}. 
The key difference between this model and ours is that, in this model, the structure of the solution can differ significantly from that of the nominal problem, whereas in our setting, the structure of a feasible solution remains~unchanged.	
	\section{Preliminaries}\label{sec:preliminaries}
We consider simple undirected graphs consisting of a set of nodes denoted by $V(G)$ and a set of edges denoted by $E(G)$.
Occasionally, we omit the notation $G$ and instead only write $V$ or $E$.
A \emph{stable set}, widely also known as an \emph{independent set} in a graph $G$ is a subset of vertices $V'\subseteq V$ such that the induced subgraph $G[V']$ does not contain any edge.
A \emph{matching} in a graph consists of a set of pairwise-disjoint edges.
A graph $G$ is {\em bipartite} if the vertex set can be partitioned into two disjoint sets $V = A\ \dot\cup\ B$ such that each edge $e \in E$ has one endpoint in $A$ and the other in $B$.
An {\em interval graph} is a graph that has an interval representation: each vertex $v \in V$ can be associated with an interval $i_v$ on the real line and two vertices share an edge if and only if their corresponding intervals intersect.
The \emph{line graph} of a graph $G$ contains the vertex set corresponding to $E(G)$, and there is an edge between two vertices if and only if the corresponding edges share an endpoint in~$G$.

A non-empty downward-closed set system $(E, \mathcal{I})$ with $\mathcal{I}\subseteq 2^E$ is a \emph{matroid} if  it satisfies the \emph{exchange property}:
\begin{equation}\label{ex1} 
	\mbox{if } I,J \in \mathcal{I} \mbox{ and } |I|<|J|, \mbox{ then } I+g \in  \mathcal{I} \mbox { for some } g\in J\setminus{I}.
\end{equation}

Given a matroid $\mathcal{M}=(E,\mathcal{I})$, a subset $I\subseteq E$ is called \emph{independent} if $I\in \mathcal{I}$, and \emph{dependent} otherwise. 
For $X\subseteq E$, an inclusion-wise maximal independent subset of $X$ is called a \emph{basis} of~$X$. 
Let us denote by $\mathcal{B}$
the collection of all bases of $E$ which are also the bases of the matroid $\mathcal{M}$.
It follows from (\ref{ex1}) that all bases of $\mathcal{M}$ have the same cardinality.  
For an element $\f \in E$, we define $\mathcal{M}-\f \coloneqq (E - \f, \mathcal{I}')$, where $\mathcal{I}' \coloneqq \{ I \in \mathcal{I} \mid \f \notin I \}$.

It is a well-known fact \cite[Thm.~2]{brualdi1969comments} that a non-empty collection $\mathcal{B}\subseteq 2^E$ of sets forms the base set of a matroid if and only if the following \emph{strong basis exchange property} holds:
\begin{align*} 
	& \mbox{if } B_1, B_2 \in \mathcal{B},\ a \in B_1 \setminus{B_2}, \mbox{ then  } \exists  b \in B_2 \setminus{B_1} \tag{2} \\
	& \mbox{with } B_1-a+b \in \mathcal{B} \mbox{ and } B_2-b+a \in \mathcal{B}.
\end{align*}

	\section{Robust Matroid}\label{sec:robust-matroids}

In this section, we consider the problem $(k, \ell)$-\robustmatroidbasis\ ($(k,\ell)$-\robustmatroidbasisshort) for $k \leq \ell$, where we are given a matroid  $\mathcal{M}=(E, \mathcal{I})$ with weight function $w \colon E \rightarrow \mathbb{R}$, and the task is to solve \eqref{eq: (k,l)recoverablerobustproblem} for $k \leq \ell$ with the feasibility set $\mathcal{F}$ equal to the independent set system~$\mathcal{I}$.
That is, in this section, we consider the following problem:
\begin{problemenv}
	\problemtitle{$(k, \ell)$-\robustmatroidbasis ($(k, \ell)$-\robustmatroidbasisshort)}
	\probleminput{A matroid  $\mathcal{M}=(E, \mathcal{I})$ with weight function $w \colon E \rightarrow \mathbb{R}_+$.}
	\problemquestion{Find $I\in \mathcal{I}$ such that for all $\F\subseteq E, |\F|\leq k$, there exists 
		$\Rem \subseteq E\setminus \F, |\Rem|\leq \ell$ such that $(I\setminus \F)\cup \Rem\in \mathcal{I}$, and the weight $w((I\setminus \F)\cup \Rem)$ is maximized.}
\end{problemenv}

A \emph{circuit} of a matroid is an inclusion-wise minimal dependent set. 
Recall that for any basis $B\in \mathcal{B}$ of a matroid $\mathcal{M}$, and any $g\in E\setminus{B}$, there exists a unique circuit, called $C(B,g)$, among the elements in  $B\cup \{g\}$. 
Recall further that the collection of circuits $\mathcal{C}$ of a matroid satisfies the following  \emph{cycle exchange property}: 
for any two  circuits $C_1, C_2\in \mathcal{C}$, and any $g\in C_1\cap C_2$, there exists  a circuit $C_3\in \mathcal{C}$ with $C_3\subseteq (C_1\cup C_2) \setminus{\{g\}}$.

Furthermore, matroids are closed under  deletion and contraction operations. That is, for any matroid $\mathcal{M} = (E, \mathcal{I})$ and any $\F \subseteq E$, the system $\mathcal{M} \setminus \F = (E \setminus \F, \mathcal{I} \setminus \F)$ with independence system $\mathcal{I} \setminus \F \coloneqq \{I \subseteq E \setminus \F \mid I \in \mathcal{I}\}$, as well as the system $\mathcal{M} /F = (E \setminus \F, \mathcal{I} / F)$ with independence system $\mathcal{I} / F \coloneqq \{I \subseteq E \setminus \F \mid F\cup I \in \mathcal{I}\}$,  are again matroids.
For more details see, e.g., \cite[Chapter~39]{schrijver2003combinatorial}.

Let $\mathcal{B}$ be the collection of bases of the matroid $\mathcal{M}=(E,\mathcal{I})$. We will prove that for any  weight function $w \colon E \rightarrow \mathbb{R}$, any maximum-weight basis in $\mathcal{B}$ is also an optimal solution to $(k,\ell)$-\robustmatroidbasisshort\ if $k \leq \ell$.
We use the following lemma.

\begin{restatable}[]{lemma}{matroidhelper}
	\label{lem:matroid-helper}
	Let $B$ be a maximum-weight basis of $\mathcal{M}$ and let $\f \in B$.
	Further, let $\e \in \argmax \{ w(g) \mid g\in E\setminus{\{\f\}}, \ B - \f + g \in \mathcal{B} \}$.
	Then $B - \f + \e$
	is a maximum-weight basis of~$\mathcal{M} - \f$.
\end{restatable}
\begin{proof}
	Let  $B' = B - \f + \e$,
	and let $B^\star$ be a maximum weight independent set in $\mathcal{M}-\f$ with $|B^\star \cap B'|$ maximum.
	If $B' \neq B^\star$, take any $a \in B^\star \setminus B'$.
	By the strong basis exchange property there exists $b \in B' \setminus B^\star$ with $B^\star - a + b \in \mathcal{B}$ and $B' -b + a \in \mathcal{B}$.
	Therefore we have that (i) $w(a) > w(b)$, since $B^\star$ is a maximum weight basis of $\mathcal{M}-\f$,
	and by our choice of $B^\star$ being closest to $B'$.
	Hence, $B - b + a$ cannot be a basis of $\mathcal{M}$, since $w(B - b + a) > w(B)$, and $B$ was a maximum weight basis of $\mathcal{M}$.
	Therefore we have that $b \notin C(B, a)$, but $b \in C(B', a)$,  where $B'= B -\f+\e$.
	This implies that $\f \in C(B, a)$, and thus (ii) $w(a) \leq w(\e)$ by the choice of~$\e$ as an element of maximum weight among those that can be feasibly exchanged with $f$.
	But then (i) and (ii) together imply that $w(b) < w(\e)$. Since $\f\in C(B,\e)\cap C(B,a)$, by the circuit exchange property, there exists a circuit  $\bar{C}\subseteq \left( C(B,\e)\cup C(B,a) \right) \setminus\{\f\}\subseteq B-\f+\e+a=B'+a$. Thus, this circuit $\bar{C}$ must be equal to the unique circuit $C(B',a)$ which contains $b$. 
	It follows that $b\in C(B',a)=\bar{C} \subseteq \left( C(B,\e)\cup C(B,a) \right) \setminus\{\f\}$ with $b\not\in C(B,a)$, implying that $b\in C(B,\e)$. Summarizing, $B-b+\e$ is a basis of larger weight in matroid $\mathcal{M}$ than $B$, a contradiction.
\end{proof}

\begin{restatable}[]{theorem}{matroidbasisopt}
	\label{thm:matroid:basis-opt}
	Any maximum-weight basis of $\mathcal{M}$ is an optimal first-stage solution for $(k,\ell)$-\robustmatroidbasis\ if $k \leq \ell$. Thus, $(k,\ell)$-\robustmatroidbasis is polynomial-time~solvable if $k \leq \ell$.
\end{restatable}    
\begin{proof}
	Let $S$ be a maximum weight basis of $\mathcal{M}$ and let $S^\star$ be an optimal first-stage solution to \robustmatroidbasisshort.
	Furthermore, let $\F$ be the set of elements that are deleted by the adversary in the worst-case for the first-stage solution $S$ and let $S'$ be its second-stage solution after the deletion of $\F$ and the recourse. 
	Thus, $S'$ is a maximum-weight basis in matroid $\mathcal{M} \setminus{F}/(S\setminus{F})$.
	Additionally, let $S^\star_\F$ be the second-stage solution for the first-stage solution $S^\star$ with interdiction set $\F$.
	We will show that $w(S') \geq w(S^\star_\F)$. 
	Since $\F$ was the worst-case choice of the adversary for $S$ and $\F$ could also be chosen as an interdiction set for $S^\star$, we have that then $S$ is also an optimal first-stage solution.
	Hence, it remains to prove $w(S') \geq w(S^\star_\F)$. 
	
	Consider the adversary and the recourse as a one-by-one process, in which the elements in~$\F$ are deleted one after the other in a fixed order, and after each deletion of an element (the adversary's action) another element is added to the solution (the recourse), until all elements of~$\F$ have been deleted and each recourse has been added.
	That is, we have a fixed order of the elements in~$\F$, i.e., $\f_1, \f_2, \dots, \f_k$, $\f_i \in \F$ for $1 \leq i \leq k$ and in step $i$ the element $\f_i$ is deleted and an element $e_i \in E \setminus (S \cup \{\f_1, \f_2, \dotsc, \f_i\})$ 
	is added to $S$ such that $(S \cup \{\e_1, \e_2, \dots, \e_i\}) \setminus \bigcup_{j = 1}^{i} \{\f_j\}$ is an independent set (or an element $\e_i^\star \in E \setminus (S^\star \cup \{\f_1, \f_2, \dots, \f_i\})$ for $S^\star$, respectively).
	Hence, we consider this process for $S$ and $S^\star$.
	Note that we can do this process if $k \leq \ell$.
	
	In this process, one can potentially add elements from $(E \setminus S) \cap \F$ that are then removed again in a later step (and the same for $S^\star$).
	Furthermore, in this process, as recourse for the elements in the first-stage solution $S$ we always add the elements from the remaining set of elements $E \setminus (S \cup \{\f_1, \f_2, \dots, \f_{i-1}\})$ that maximizes the weight of the resulting basis of $\mathcal{M} \setminus \bigcup_{j = 1}^{i-1} \{\f_j\}$.
	That is, $\e_i \in \argmax \{ w(g) \mid \big( (S  + g - \f_i \cup \{\e_1, \e_2, \dots, \e_{i-1}\}) \setminus \bigcup_{j = 1}^{i} \{\f_j\} \big) \in \mathcal{I} \}$. 
	Note that this is always possible due to the strong basis exchange property  (taking
	$B_1:= S\cup\{e_1, \ldots, e_{i-1}\} \setminus{\{f_1, \ldots, f_{i-1}\}}$, $ a:=f_j$, and $B_2:=S'$).
	
	For the optimal first-stage solution $S^\star$, if some element is deleted during the process, we simply add some element from the final solution $S^\star_\F$.
	Let $S_i$ and $S_i^\star$, $0 \leq i \leq k$, be the intermediate bases obtained by this process for $S$ and $S^\star$, respectively.
	Again, note that we can always add an element for each removed element since $k \leq \ell$.
	
	Now we invoke Lemma~\ref{lem:matroid-helper} and have that $S_1$ is a maximum weight basis of $\mathcal{M}-\f_1$. 
	Furthermore, iteratively we can now invoke Lemma~\ref{lem:matroid-helper} to show that $S_i$ is a maximum weight basis of $\mathcal{M} \setminus \bigcup_{j = 1}^{i} \{\f_j\}$.
	Hence,  we have that $w(S_i) \geq w(S^\star_i)$ for all $1 \leq i \leq k$.
	In particular,  since $S_k$ is an optimal basis of $\mathcal{M}\setminus{F}$, and since $S_k$ contains all elements from $S\setminus{F}$,  it follows that $w(S')=w(S_k)$, implying $w(S') = w(S_k) \geq w(S_k^\star) = w(S^\star_\F)$.
	Consequently, by using the polynomial-time greedy algorithm for the matroid basis problem \cite{kruskal1956shortest}, we can compute an optimal first-stage solution for $(k,\ell)$-\robustmatroidbasis in polynomial time if $k \leq \ell$.
\end{proof}

We observe that this result does not hold if $k > \ell$, which is highlighted in the following example.
It remains open whether this problem admits a polynomial-time algorithm for $k > \ell$.

\begin{figure}[t]
	\begin{center}
		\begin{tikzpicture}[scale=1]
			\begin{scope}
				\node[small vertex, label={left:$c$}] (1) at (0,0) {};
				\node[small vertex, label={left:$a$}] (2) at (0,2) {};
				\node[small vertex,label={right:$b$}] (3) at (2,2) {};
				\node[small vertex, label={right:$d$}] (4) at (2,0) {};
				\draw (1)-- node[left]{$5$} (2);
				\draw (2)-- node[above]{$3$}(3);
				\draw (3)-- node[right]{$3$}(4);
				\draw (1)-- node[below]{$8$}(4);
				\draw (1)-- node[below right]{$4$}(3);
			\end{scope}
			\begin{scope}[xshift=4cm]
				\node[small vertex, label={left:$c$}] (1) at (0,0) {};
				\node[small vertex, label={left:$a$}] (2) at (0,2) {};
				\node[small vertex,label={right:$b$}] (3) at (2,2) {};
				\node[small vertex, label={right:$d$}] (4) at (2,0) {};
				\draw[line width=2pt] (1)-- node[left]{$5$} (2);
				\draw (2)-- node[above]{$3$}(3);
				\draw (3)-- node[right]{$3$}(4);
				\draw[line width=2pt] (1)-- node[below]{$8$}(4);
				\draw[line width=2pt] (1)-- node[below right]{$4$}(3);
			\end{scope}
			\begin{scope}[xshift=8cm]
				\node[small vertex, label={left:$c$}] (1) at (0,0) {};
				\node[small vertex, label={left:$a$}] (2) at (0,2) {};
				\node[small vertex,label={right:$b$}] (3) at (2,2) {};
				\node[small vertex, label={right:$d$}] (4) at (2,0) {};
				\draw[line width=2pt] (1)-- node[left]{$5$} (2);
				\draw[line width=2pt] (2)-- node[above]{$3$}(3);
				\draw[line width=2pt] (3)-- node[right]{$3$}(4);
				\draw (1)-- node[below]{$8$}(4);
				\draw (1)-- node[below right]{$4$}(3);
			\end{scope}
		\end{tikzpicture}
	\end{center}
	\caption{Counterexample showing that an optimal nominal solution (middle) need not be optimal for the robust counterpart of the $(2,1)$-\robustmatroidbasis problem.
	}
	\label{fig:rroc:matroid:counterexample}
\end{figure}

\begin{example}\label{example:rroc:matroid:counterexample}
	Consider the instance of $(2,1)$-\robustmatroidbasis as illustrated in Figure~\ref{fig:rroc:matroid:counterexample}, in the left picture. The maximum-weight independent sets, i.e., the bases of the (graphic) matroid 
	are the spanning trees of the depicted graph. 
	It is not hard to see that the unique maximum-weight basis, which is the spanning tree shown in the middle with a total weight $17$, is not an optimum solution for the robust counterpart: If we start with the spanning tree of value $17$ as 
	first-stage solution, the guaranteed surviving value is only $7$, since the adversary might delete the two edges of weight $5$ and $8$, and we can only augment by adding one of the two edges of value $3$. 
	Instead, if we start with the spanning tree consisting of the edges of values $3, 5$ and $3$ as shown in the right picture, the adversary will delete edges with value $8$ and $3$ (vertical). Without adding any other edge in the recourse step, the total weight is $8$.
\end{example}

Further, we show that matroids are the only downward-closed set systems for which an optimal solution to the nominal problem is also an optimal solution to its robust counterpart.

\begin{theorem}
	\label{thm:non-matroids}
	For every non-empty downward-closed set system $\mathcal{F} \subseteq 2^E$ which is not a matroid, there exist weights $w \colon E \to \mathbb{R}_+$  such that an optimal solution to the nominal problem is not an optimal first-stage solution for the robust~counterpart.
\end{theorem}

Recall that a non-empty downward-closed system $\mathcal{F}\subseteq 2^E$ is not a matroid  if the exchange property (\ref{ex1}) is violated. Such \emph{non-matroids} can be characterized via the following Lemma (see, e.g., \cite{fujishige2017matroids} for the proof).

\begin{lemma}\label{l.non-matroids}
	Let $\mathcal{F}\subseteq 2^E$ be a non-empty downward-closed system which is not a matroid, and let $\bar{\mathcal{F}}\subseteq \mathcal{F} $ be the collection of the inclusion-wise maximal sets  in $\mathcal{F}$. Then there exist two sets $X,Y\in \bar{\mathcal{F}}$, and three elements $\{a,b,c\}\subseteq X\Delta Y \coloneqq (X\setminus{Y})\cup (Y\setminus{X})$ such that for each set $Z\in \bar{\mathcal{F}}$ with $Z\subseteq X\cup Y$, either $a\in Z$, or $\{b,c\}\subseteq Z$. 
\end{lemma}
\begin{proof}[Proof of Theorem \ref{thm:non-matroids}.]
	Take $X,Y\in \bar{\mathcal{F}}$, and the three elements $\{a,b,c\}\subseteq X\Delta Y$ as in the statement of Lemma \ref{l.non-matroids}.
	We set the weights of  all elements in $E\setminus{\{a,b,c\}}$ to zero,  and  $w(a) = 3$, $w(b) = w(c) = 2$.
	Then the optimal first-stage solution for the $(1, 1)$-robust counterpart is to select either $S=\{b\}$, or $S=\{c\}.$ Then, if the adversary interdicts, one can still add $\{a \}$ and achieve a second-stage value of 3.
	Otherwise, if the adversary does not interdict, one can extend to the optimal solution of the underlying problem $\{b, c\}$ of value 4.
	
	In contrast, if the first-stage solution $S$ is equal to a  nominal optimal solution, then $S$ necessarily contains $\{b, c \}$, and  one can only achieve a value of $2$. 
	Because if either $b$ or $c$ is interdicted, one cannot extend $S$ by an element of positive value in the recourse. 
\end{proof}

	\section{Robust Stable Set in Interval Graphs}
\label{sec:robust-interval-stable-set}
In this section we consider \robuststableset in interval graphs.
For convenience, we use the standard interval representation of an interval graph $G=(V,E)$ where each node is represented as an interval in $\mathbb{R}$, and a stable set in $G$ corresponds to a selection of  pairwise-disjoint intervals. 
We refer to the task of finding a maximum-weighted set of disjoint intervals as \intervalselection (\intervalselectionshort) and its robust counterpart as \robustintervalselection (\robustintervalselectionshort). 
Further, we refer to the intervals as \emph{jobs} as it is common in scheduling.
It is known that \intervalselectionshort is solvable in polynomial time via dynamic programming \cite{zbMATH03513839}.
We formally define the \robustintervalselection problem as follows.

\begin{problemenv}
	\problemtitle{(1,1)-\robustintervalselection (\robustintervalselectionshort)}
	\probleminput{A set of intervals $I = \{i_1, i_2, \dots, i_n\}$, where $i_j = [a_j, b_j)$, $a_j < b_j$, $a_j, b_j\in \mathbb{N}$, a weight function $w \colon I\to \mathbb{R}_{\geq 0}$, and a family $\mathcal{F}\subseteq 2^{I}$ of sets of pairwise-disjoint intervals in $I$.}
	\problemquestion{Find a set $S \in \mathcal{F}$ which maximizes the total weight $w(S - i_j + i_\ell)$ for the deletion of any interval $i_j$ and the best possible recourse by another interval $i_\ell$ thereafter.
	}
\end{problemenv}

\noindent Note that $\{\emptyset\}\in I$, since we redefine each of our ground sets such that it includes~$\{\emptyset\}$;
this way the empty set in $I$ can be used to show the possible empty deletion and empty recourse. 
We prove the following main result. 

\begin{theorem}
	\label{thm:interval-selection-poly}
	There is a polynomial-time algorithm for $(1,1)$-\robustintervalselection.
\end{theorem}

Later in Subsection \ref{sec:generalization-scheduling}, we show how to generalize this theorem to $(1,p)$-\robustintervalselection for any constant $p\in \mathbb{N}$, in which we can add up to $p$ intervals as recourse after the deletion of a single interval.

\subsection{Algorithm Outline and Proof Roadmap}
Here, we outline our algorithmic approach for the proof of Theorem~\ref{thm:interval-selection-poly}.
As mentioned in the introduction, one can reduce the robust counterpart of a nominal problem to its bounded-regret version~(\ref{eq:opt:sol:bounded-regret}).
The bounded-regret version of \robustintervalselectionshort, which we denote as \boundedregretintervalselection (\boundedregretintervalselectionshort), asks to find $S \in \mathcal{F}$ which optimizes the following objective function:
\begin{align}\label{eq:opt:sol:RIS-bounded-regret}
	\max_{\substack{S\in \mathcal{F}\\ \Delta(i_j, S) \leq \lambda \colon\forall i_j \in I }} ~w(S),  \tag{$\lambda$-RIS}
\end{align}

\noindent
where $\Delta(i_j, S) = w(S) - \max \{w(S-i_j+i_{\ell}) \mid 	i_{\ell}\in I-i_j,  
S-i_j+i_{\ell} \in \mathcal{F}\}.$
The value $\Delta(i_j, S)$ defines the loss in the objective value occurred due to the change in the solution  
in case $i_j$ fails. 
Note that, if $i_j \in S$, then we simply have 
$\Delta(i_j, S) =  w(i_j) -   \max \left\{w(i_{\ell}) \mid i_{\ell} \in I\setminus S,\ S-i_j+i_{\ell} \in \mathcal{F} \right\}$.

Due to the established tight connection between~(\ref{eq:opt:sol:bounded-regret}) and~(\ref{eq: (1,1)recoverablerobustproblem}) formulations, a polynomial-time algorithm for \boundedregretintervalselectionshort implies  a polynomial-time algorithm for \robustintervalselectionshort, since here the number of different values of regrets is bounded by $\mathcal{O}(|I|^2)$. 
As a result, we have the following proposition.

\begin{proposition}
	\label{prop:interval-selection-bounded-regret}
	A polynomial-time algorithm for the problem \boundedregretintervalselection
	implies a polynomial-time algorithm for $(1,1)$-\robustintervalselection.
\end{proposition}

Given Proposition~\ref{prop:interval-selection-bounded-regret}, it remains to provide a polynomial-time algorithm for \boundedregretintervalselectionshort for an arbitrary fixed value of $\lambda \in \mathbb{R}$.
Recall that $\lambda$ corresponds to the maximum loss (regret) in the objective value for any job that gets deleted. 
Therefore, it will be convenient to hedge against regrets larger than $\lambda$ by defining for each job $i_j$ a backup of weight at least $w(i_j)-\lambda$;
this backup can then be used in the recourse if job $i_j$ fails. 
Further, we give a formal definition of a backup of an interval.

\paragraph{Defining backups.} Let $\mathcal{F}(I)$ be a collection of sets of pairwise-disjoint intervals in $I$, and $S\in \mathcal{F}(I)$. 
We define, for each element $i_j \in I$, a {\em backup} element $\backup(i_j, S)$: 
\[ 
\backup (i_j, S) \in \argmax_{\substack{i_\ell \in I \setminus (S + i_j)\\ S - i_j + i_\ell \in \mathcal{F}(I)}} \ w(i_\ell) \ . 
\]
Informally speaking, $\backup (i_j, S)$ is a best possible recourse that can be added to the remaining  solution $S - i_j$ when $i_j$ is interdicted.
Hence, $\backup (i_j, S)$ serves as a backup for $i_j$.

Note that, for a job $i_j\in S$, its backup  $\backup (i_j, S)=i_{j'}$ may or may not intersect with $i_j$.
Hence, the backup job $i_{j'}$ of $i_j$ satisfies one of the following two conditions, depending on whether or not $i_j$ and $i_{j'}$ intersect.
\begin{itemize}
	\item[(a)] $S + i_{j'} \notin \mathcal{F}(I)$, or
	\item[(b)] $S + i_{j'} \in \mathcal{F}(I)$.
\end{itemize}

In case (a), job $i_{j'}$ cannot be used as a backup for any other job $i \neq i_j$, since $i_j\in S$ and $i_j\cap i_{j'}\neq \emptyset$.
Hence, we call such a backup the \emph{\private} backup of job $i_j$.
In case (b), 
job $i_{j'}$ can be used as a backup for any job $i \in I\setminus \{i_{j'}\}$, 
since $i_{j'}$ does not intersect with any job in $S$, and hence we call such a backup a \emph{\universal} backup.
Note that since a \universal backup, say $u$, can be used as a backup for any job $i \in I- u$, it is enough to have one such backup for all jobs that might not have a \private backup.
For a first-stage solution set $S$, we want a \universal backup $u$ to be in $\argmax \{ w(i_{j'}) \mid i_{j'} \in I \setminus S \text{ and } S + i_{j'} \in \mathcal{F}(I) \}$. 

For the case where $\lambda<0$ and  the universal backup $u$ is interdicted, if there is no recourse, the value $\Delta(u,S)=0$ is strictly larger than $\lambda$.
This contradicts the feasibility of $S$.
To deal with this issue, we choose an \emph{extra} backup $b$ along with our universal backup $u$. In particular, we aim to have $b \in \argmax \{ w(i_{j'}) \mid i_{j'} \in I 
\setminus (S + u) \text{ and } S + i_{j'} \in \mathcal{F}(I)\}$.
Moreover, we have $w(u) \geq w(b)$. 
Note that, since we have $\{\emptyset\} \in I$, if there is no candidate job to serve as a universal backup or an extra backup, we have $u=\{\emptyset\}$ or $b =\{\emptyset\}$.
All in all, we solve an instance of the \boundedregretintervalselectionshort problem by first guessing a pair of \universal and extra backup, $u,b \in I$, and then solving the problem.

In Section~\ref{sec:interval:algorithm}, we show that for the correct guess of  \universal and extra backup, \boundedregretintervalselectionshort can be reduced to solving a more general interval scheduling problem, which we call \issh (\Issh).
Below we give a formal definition of \issh (\Issh):
An instance $\mathcal{I}=(I,w)$ of \Issh consists of a set $I$ of jobs and a weight function $w \colon I \rightarrow \mathbb{R}_{\geq 0}$. 
Each job $i_j\in I$ potentially has two half open intervals $[\ell_j, r_j)$ and $[a_j, b_j)$ satisfying $\ell_j \leq a_j < b_j \leq r_j$ associated with it. 
The inner interval $[a_j, b_j)$ is \emph{red} and the interval $[\ell_j, r_j)$ is \emph{blue}.
A feasible solution is a subset $S$ of jobs such that for any two jobs $i_j, i_k \in S$, we have that $[a_j, b_j)$ and $[\ell_k, r_k)$ do not intersect and that $[a_k, b_k)$ and $[\ell_j, r_j)$ do not intersect. 
That is, the red interval of any job in the solution is not allowed to intersect with the 
blue interval of any other job in the solution.	
The task is to pick a feasible solution of maximum weight. We show in Theorem~\ref{thm:interval-selection-soft-hard} in Section~\ref{sec:interval:algorithm} that this problem can be solved in polynomial time.

\subsection{The Algorithm for Bounded-Regret Interval Scheduling}
\label{sec:interval:algorithm}

We present an efficient algorithm for \boundedregretintervalselection for a fixed value of $\lambda$.
\begin{restatable}[]{theorem}{boundedregretinterval}
\label{thm:IS:helper}
Given a polynomial-time algorithm for \issh, there is a polynomial-time algorithm for \boundedregretintervalselection. 
\end{restatable}
\begin{proof}
To prove the theorem, we first describe the reduction from \boundedregretintervalselection to \issh(\Issh) and then prove its correctness.

Let $\mathcal{I}=(I,w)$ be an instance of \boundedregretintervalselectionshort. 
In the first step, we enumerate over all pairs of jobs $u, b \in I$~ such that $w(u) \geq w(b) \geq - \lambda$, and corresponding to each such pair, we construct an instance $\mathcal{I}_{u, b}'(\lambda)=(I_{u, b}'(\lambda), w')$ of \Issh.

First, fix a pair $u,b\in I$ corresponding to the universal backup and the \extra backup, respectively, each of which may be empty. 
Note that there are three cases: (i) both a \universal and \extra backup exist, (ii) an \extra backup exists, but a \universal backup does not, or (iii) neither a \universal nor an \extra backup exists.

For each job $i_j\in I$, we add the following set of jobs in $I_{u,b}'$ with weights $w'$, which cover all possible recourse actions for $i_j$. 
We further define a function $\lambda_{u, b}$ to model the regret and use this function to prune $I_{u,b}'$ in order to obtain $I_{u, b}'(\lambda)$.

\begin{itemize}
	\item The first set of jobs covers the possibility of using the \universal backup $u$ as recourse. 
	Thus, for each job $i_j \in I \setminus \{ u, b \}$ in \boundedregretintervalselectionshort, we add a job $i'_j$ to $I_{u, b}'$ in \Issh with $w'(i'_j) = w(i_j)$. 
	Moreover, the red and blue interval of $i'_j$ is set to $[a_j, b_j)$, which is the interval of $i_j$. 
	We set $\lambda_{u, b}(i'_j) = w(i_j) - w(u)$.
	\item The second set of jobs corresponds to the possibility of using a \private backup. 
	Thus, for each ordered pair of jobs $(i_j, i_k)$ from set  $I \setminus \{ u, b \}$, $i_j \neq i_k$, such that $i_j$ and $i_k$ intersect, we add a job $i^k_j$ to $I_{u, b}'$.
	We set $w'(i^k_j) = w(i_j)$ and $\lambda_{u, b}(i^k_j) = w(i_j) - w(i_k)$.
	Let $\ell_{j,k} \coloneqq \min \{a_j, a_k \}$ and $r_{j,k} \coloneqq \max \{ b_j, b_k \}$.
	Then we define the red interval of $i^k_j$ to be $[a_j, b_j)$ and the blue interval to be $[\ell_{j,k}, r_{j,k})$. 
\end{itemize}

\noindent\textit{Pruning:} Finally we want to make sure that the red interval of a job does not intersect with the jobs corresponding to the universal and extra backups.
For that, we simply delete  all those jobs from $I'_{u,b}$ which intersect with $u$ and $b$.
Further, we define $I'_{u, b}(\lambda) \coloneqq \{ i'_j \in I'_{u, b} \mid \lambda_{u, b}(i'_j) \leq \lambda \} $
to include all jobs of $I'_{u, b}$ that have a $\lambda_{u, b}$-value of at most $\lambda$.
As a result we get our instance $\mathcal{I}'_{u, b}(\lambda)= (I'_{u, b}(\lambda), w')$, here $w'$ is restricted to the set $I'_{u, b}(\lambda)$. 
See Figure \ref{fig:rroc:ISto:ISC} for reference.
Let $\mathcal{F}(I'_{u, b}(\lambda))$ be the set of feasible solutions of $\mathcal{I}'_{u, b}(\lambda)$.

\begin{figure}[t]
	\begin{center}		
		\begin{minipage}[b]{0.44\linewidth}
			\begin{ichart}[0.62,yscale=1.2]{5}{10}
				\intervalnominal{1}{4}{6}{$1$}{$i_1$}
				\intervalnominal{2}{3.5}{7.5}{$2$}{$i_2$}
				\intervalnominal{3}{1.5}{3}{$10$}{$i_3$}
				\intervalnominal{3}{8}{10}{$10$}{$i_4$}
				\intervalnominal{4}{0}{7}{$13$}{$i_5$}
				\intervalnominal{5}{6}{9}{$14$}{$i_6$}
				\draw[black!40, thick, dashed, rounded corners] (-.7,0) -- (10.7, 0) -- (10.7,-4) -- (-.7, -4) -- cycle;
			\end{ichart}
		\end{minipage}
		\hfill
		\begin{tikzpicture}
			\node at (0, 0) {$\longrightarrow$};
			\node at (0, -2) {};
		\end{tikzpicture}
		\hfill
		\begin{minipage}[b]{0.44\linewidth}
			\begin{ichart}[0.62,yscale=1.2]{10}{10}
				\intervaliscblue{1}{1.5}{3}{}{}
				\intervaliscblue{2}{8}{10}{}{}
				\intervaliscblue{3}{0}{7}{}{}
				\intervaliscblue{4}{6}{9}{}{}
				\intervaliscblue{5}{0}{7}{}{}
				\intervaliscblue{6}{6}{10}{}{}
				\intervaliscblue{7}{0}{7}{}{}
				\intervaliscblue{8}{6}{10}{}{}
				\intervaliscblue{9}{0}{9}{}{}
				\intervaliscblue{10}{0}{9}{}{}
				\intervalnominal{1}{1.5}{3}{$10$}{$i_3'$}
				\intervalnominal{2}{8}{10}{$10$}{$i_4'$}
				\intervalnominal{3}{0}{7}{$13$}{$i_5'$}    				\intervalnominal{4}{6}{9}{$14$}{$i_6'$}
				\intervalnominal{5}{1.5}{3}{$10$}{$i_3^5$}
				\intervalnominal{6}{8}{10}{$10$}{$i_4^6$}    			\intervalnominal{7}{0}{7}{$13$}{$i_5^3$}    			\intervalnominal{8}{6}{9}{$14$}{$i_6^4$}
				\intervalnominal{9}{0}{7}{$13$}{$i_5^6$}
				\intervalnominal{10}{6}{9}{$14$}{$i_6^5$}
				\draw[black!40, thick, dashed, rounded corners] (-.7,0) -- (10.7, 0) -- (10.7,-7) -- (-.7, -7) -- cycle;
				\node () at (9.3,-0.2) {\scriptsize $\lambda = -1$};
				\node () at (9.12,-0.4) {\scriptsize $u = i_2$};
				\node () at (9.17,-0.62) {\scriptsize $b = i_1$};
				\node () at (2.2,-.6) {$\tikzxmark$};
				\node () at (9,-1.2) {$\tikzxmark$};
				\node () at (3.5,-1.8) {$\tikzxmark$};
				\node () at (7.5,-2.4) {$\tikzxmark$};
				\node () at (3.5,-4.2) {$\tikzxmark$};
				\node () at (7.5,-4.8) {$\tikzxmark$};
				\node () at (3.5,-5.4) {$\tikzxmark$};
				\node () at (7.5,-6) {$\tikzxmark$};
			\end{ichart}
		\end{minipage}
	\end{center}
	\caption{Reduction from $(1,1)$-\robustintervalselection to \issh. Here, we guessed $\lambda = -1$, and $i_2$ as the universal and $i_1$ as the extra backup. 
		The crossed jobs in the instance of \Issh are the pruned jobs.
		The jobs $i_5', i_6', i_5^3, i_6^4, i_5^6$, and $i_6^5$ are pruned because their red part overlaps with the universal backup.
		Additionally, the jobs $i_3'$ and $i_4'$ are pruned because their $\lambda'$ values are greater than the current $\lambda =-1$ value.}
	\label{fig:rroc:ISto:ISC}
\end{figure}

In the remainder of the proof, we show that an optimal solution $S^\star$ of $\mathcal{I}$ is in one-to-one correspondence with the optimal solution  $S$ of $\mathcal{I}'_{u, b}(\lambda)$ which maximizes the value $w'(S)$ over all  $u, b \in I$. 
Moreover, the set $S$ has the same weight and regret value as that of $S^\star$.
To get an intuition, first observe that all jobs in $I'_{u, b}(\lambda)$ correspond to the jobs in $I$; we simply add some blue intervals to them. 
Blue intervals of jobs correspond to their respective backup jobs, and hence they are allowed to intersect in a chosen solution of \Issh. 
Furthermore, note that for all jobs $i_j' \in I'_{u, b}(\lambda)$ we have $\lambda_{u, b}(i'_j) \leq \lambda $, which encodes that the regret we suffer if the corresponding interval $i_j$ is interdicted is indeed at most $\lambda$.
Now, it remains to show that from an optimum solution $S^\star$ of~$\mathcal{I}$, we can compute an optimum solution $S$ of $\mathcal{I}'_{u, b}(\lambda)$ of the same weight, and from an optimum solution $S$ to $\mathcal{I}'_{u, b}(\lambda)$ we can compute an optimum solution $S^\star$ to $\mathcal{I}$ of the same weight and regret $\lambda$.

Let $S^\star$ be an optimum solution to $\mathcal{I}$. 
We now construct in polynomial time a solution to the corresponding instance $\mathcal{I}'_{u, b}(\lambda)$ of \issh of the same weight.
By definition we have that $w(i_j) - w( \backup (i_j, S^\star)) \leq \lambda$. Without loss of generality, we assume that $S^\star$ does not contain $u$ and $b$. 

By construction of $\mathcal{I}'_{u, b}(\lambda)$, for each $i_m \in S^\star$ that uses the \universal backup $u$ when $i_m$ fails, there is a new job $i'_m$ in $I_{u, b}'(\lambda)$ with $w'(i_m') = w(i_m)$ and $\lambda_{u, b}(i'_m) = w(i_m) - w(u)$. 
Since $ w(i_m) - w(\backup (i_m, S^\star)) \leq \lambda$, 
we have that $\lambda_{u, b}(i'_m) \leq \lambda$, and hence $i'_m \in I'_{u, b}(\lambda)$. 
We let 
\[ 
S_U \coloneqq \{ i'_m \mid i_m \in S^\star \text{ and $i_m$ has \universal backup} \} \ .
\]

Now consider any job $i_m \in S^\star$ that has a \private backup $i_k \coloneqq \backup (i_m, S^\star)$.
By construction of~$I'_{u, b}(\lambda)$, there is a job $i^k_m \in \mathcal{I}'_{u, b}(\lambda)$ with $w'(i^k_m)= w(i_m)$  and  $\lambda_{u, b}(i^k_m) \leq w(i_m) - w(i_{k}) \leq \lambda$.
Hence,~$i^k_m$ is contained in $I_{u, b}'(\lambda)$.
We let 
\[
S_P \coloneqq \{ i^k_m \mid i_m \in S^\star \text{ and $i_m$ has a \private backup $i_k \coloneqq \backup (i_m, S^\star)$} \} \ .
\]

Let $S = S_P \cup S_U$.
We claim that $S$ is a feasible solution to $\mathcal{I}'_{u, b}(\lambda)$.
This directly follows from the definition of blue and red intervals: 
First, observe that all jobs in $S$ are indeed contained in~$I_{u, b}'(\lambda)$.
Second, note that no two red intervals of jobs in $S$ intersect, since $S^\star$ is a feasible solution before the failure of any job.
Third, observe that no blue interval intersects with any red interval in $S$, as
$S^\star - i_m + \backup (i_m, S^\star)$ is feasible for any job $i_m \in S^\star$ and its backup $\backup (i_m, S^\star)$.

Next, we show how to construct in polynomial time a feasible solution to $\mathcal{I}$ from an optimum solution to $\mathcal{I}'_{u, b}(\lambda)$ with the same weight and regret $\lambda$.

The proof of this direction is similar to the previous direction and hence we only sketch it here.
From an optimal solution $S$ to $\mathcal{I}'_{u, b}(\lambda)$ one can obtain in a straightforward way a solution $S^\star$ to $\mathcal{I}$.
It is then easy to see that the regret of $S^\star$ is at most $\lambda$ and the weights of the solutions $S$ and $S^\star$ coincide.

This concludes the proof.
\end{proof}

Finally, we provide a dynamic program that solves \issh in polynomial time. 

\begin{restatable}[]{theorem}{thmISSHpolytime}
\label{thm:interval-selection-soft-hard}
There is a polynomial-time algorithm for \issh.
\end{restatable}

\begin{proof}
We extend the classical dynamic programming algorithm for \intervalselection~\cite{zbMATH03513839} to \Issh.
We denote by $\mathcal{F}(I)$ the set of feasible solutions for the set of jobs in $I$, where $I$ is an instance of \issh. In particular, any $S \in \mathcal{F}(I)$ satisfies that for any job in~$S$, the red part of that job does not overlap with the red or blue part of any other job in~$S$.
Let $A = \{ a_1, \dots, a_n \}$ and $B = \{ b_1, \dots, b_n \}$ be the sets of start and end points of all red intervals of the jobs, respectively.
Furthermore, let $L = \{ \ell_1, \dots, \ell_n \}$ and $R = \{ r_1, \dots, r_n \}$ be the sets of start and end points of all blue intervals of the jobs, respectively.
These sets $A$, $B$, $L$, and $R$ need not be ordered.
We define $T_b = \max_{1 \leq i \leq n} b_i$ and $T_r = \max_{1 \leq i \leq n} r_i$.

For all $t_b \in B$ and for all $t_r \in R$, define $B_{t_b, t_r} = \{i_j \in I \mid b_j = t_b \text{ and } r_j \leq t_r \}$ and $R_{t_b, t_r} = \{i_j \in I \mid r_j = t_r \text{ and } b_j \leq t_b \}$.
Additionally, let $I_{t_b, t_r} = \{ i_j \in I \mid b_j \leq t_b \text{ and } r_j \leq t_r \}$.
Analogous to $\mathcal{F}(I)$ we also define $\mathcal{F}(I_{t_b, t_r})$. 

In the algorithm, for each tuple $(t_b, t_r)\in B\times R$,   we compute a maximum weight set $S \in \mathcal{F}(I_{t_b, t_r})$. 
We define
\[ 
    S^\star(t_b, t_r) \coloneqq \argmax_{S \in \mathcal{F}(I_{t_b, t_r})} w(S) \quad \text{ and } \quad s^\star(t_b, t_r) \coloneqq \max_{S \in \mathcal{F}(I_{t_b, t_r})} w(S) \ .
\]

Additionally, for $t, t' \in A \cup B \cup L \cup R$, let $s^\star(t, t') \coloneqq \max_{t_b \leq t, t_r \leq t'} s^\star(t_b, t_r)$ (and similarly $S^\star(t, t')$). 
Our goal is to find $S^\star(T_b, T_r)$ and $s^\star(T_b, T_r)$.
Suppose that we have computed the following:
\begin{itemize}
	\item $S^\star(t, t')$ and $s^\star(t, t')$ for all $t \leq \tau_b$ and $t' \leq \tau_r$, $t, t' \in A \cup B \cup L \cup R$ for some $\tau_b \in B$ and $\tau_r \in R$ such that $\tau_b \leq \tau_r$, and 
	\item  $S^\star(t, t')$ and $s^\star(t, t')$ for all $t < \tau_b$ and $t, t' \in A \cup B \cup L \cup R$. 
\end{itemize}

For the first step, let $t'_r \coloneqq \min \{t'_r \in R \mid t'_r > \tau_r \}$ be the next time point after $\tau_r$ at which a blue interval of some job ends.

We get
\[ 
s^\star(\tau_b, t'_r) = \max \Big\{ s^\star(\tau_b, \tau_r), ~ \max_{i_j \in R_{\tau_b, t'_r}} \{ s^\star(\ell_j, a_j) + w(i_j) \}  \Big\} \  
\]

\noindent
and $S^\star(\tau_b, t'_r)$ is the corresponding set where $s^\star(\tau_b, t'_r)$ is attained.
After exhaustively computing this first step for all $t_r'\in R$, we can compute $S^\star(t, t')$ and $s^\star(t, t')$ for all $t \leq \tau_b$ and $t' \in A \cup B \cup L \cup R$.

For the second step, let $t'_b \coloneqq \min \{t'_b \in B \mid t'_b > \tau_b \}$ be the next time point after $\tau_b$ at which a red interval of some job ends.
Furthermore, let $r_{t'_b} \coloneqq \min \{ r \in R \mid \exists i_j \in I \text{ such that } b_j = t'_b \text{ and } r_j = r \}$
and define 
$I^=_{t_b, t_r} = \{i_j \in I \mid b_j = t_b \text{ and } r_j = t_r \}$. Note that $I^=_{t'_b, r_{t'_b}} \neq \emptyset$.
We obtain $S^\star(t'_b, r_{t_b'})$ and $s^\star(t'_b, r_{t_b'})$ as follows:
\[ 
    s^\star(t'_b, r_{t_b'}) = \max \Big\{ s^\star(\tau_b, r_{t_b'}), ~ \max_{i_j \in  I^=_{t'_b, r_{t_b'}}} \{ s^\star(\ell_j, a_j) + w(i_j) \}  \Big\} \  
\]

\noindent
and $S^\star(t'_b, r_{t_b'})$ is the corresponding set where $s^\star(t'_b, r_{t_b'})$ is attained. 
Furthermore, after each of the above two steps, for $t, t' \in A \cup B \cup L \cup R$, we update $s^\star(t, t')$ and $S^\star(t, t')$.

It is straightforward to verify that the above dynamic programming algorithm produces an optimal solution in polynomial time, which proves the theorem.
\end{proof}

\bigskip

Combining Proposition~\ref{prop:interval-selection-bounded-regret}  with Theorem~\ref{thm:IS:helper} and Theorem~\ref{thm:interval-selection-soft-hard} directly implies Theorem~\ref{thm:interval-selection-poly}.

\subsection{Generalization to $(1,p)$-\robustintervalselection}
\label{sec:generalization-scheduling}

In this section, we consider the $(1,p)$-\robustintervalselection problem ($(1,p)$-\robustintervalselectionshort), a natural generalization of $(1,1)$-\robustintervalselection, in which the adversary still deletes one element but the algorithm may recover by adding up to $p$ elements.
We extend the ideas from the previous section to $(1,p)$-\robustintervalselectionshort for an arbitrary constant~$p \in \mathbb{Z}_{\geq 1}$.

\begin{theorem}
\label{thm:rros:1,p:RIS}
There is a polynomial-time algorithm for $(1,p)$-\robustintervalselection for any constant~$p \in \mathbb{Z}_{\geq 1}$.
\end{theorem}

To prove this theorem, we show that Proposition \ref{prop:interval-selection-bounded-regret} and Theorem \ref{thm:IS:helper} can be generalized.
The generalization of Proposition \ref{prop:interval-selection-bounded-regret} is straightforward, since the total number of distinct values that the regret $\lambda$ can attain is $\mathcal{O}(n^{p+1})$.
Consequently, we can again exploit the strong connection between the robust version of the problem and its bounded-regret variant.
In particular, this observation yields a polynomial-time algorithm for $(1,p)$-\robustintervalselectionshort for any constant~$p$, provided that the
bounded-regret version can be solved efficiently for each fixed value of~$\lambda$. We refer to the bounded-regret problem as $\lambda$-$(1,p)$-\robustintervalselectionshort and it remains to solve $\lambda$-$(1,p)$-\robustintervalselectionshort.
To this end, in the next paragraphs we show that $\lambda$-$(1,p)$-\robustintervalselectionshort can be reduced to a polynomial number of \issh instances, each of which being solvable in polynomial time due to Theorem~\ref{thm:interval-selection-soft-hard}.
This reduction, in turn, provides the desired generalization of Theorem~\ref{thm:IS:helper}. To do this, we first define a generalized notion of private, universal, and extra backup.

\paragraph{Defining backups for $(1,p)$-\robustintervalselection.} 
We generalize the notion of backups introduced in the last section to $(1,p)$-\robustintervalselection.
For a job $i_j$, let $P^{i_j}\coloneqq \{i_k \in I \mid [a_k, b_k) \cap [a_j, b_j) \neq \emptyset\}$ be the set of jobs that intersect $i_j$.
Let $\mathcal{F}(I)$ be a collection of sets of pairwise-disjoint intervals in $I$, and $S\in \mathcal{F}(I)$. 
We define, for each element $i_j \in I$, a {\em backup set} $\backup_p(i_j, S)$ of size at most $p$: 
\[ 
    \backup_p(i_j, S) \in \argmax_{\substack{R \subseteq I \setminus (S + i_j)\\ (S - i_j) \cup R \in \mathcal{F}(I)\\ |R| \leq p}} \ w(R) \ . 
\]
Informally speaking, $\backup_p (i_j, S)$ is a best possible recourse of size at most $p$ that can be added to the remaining solution $S - i_j$ when $i_j$ is interdicted.
Hence, $\backup_p (i_j, S)$ serves as a backup for $i_j$.
For technical reasons later we assume that this backup is unique, which can be obtained by slightly perturbing the weights of the jobs so that optimum solutions do not change.

Note that, for a job $i_j\in S$, its backup $\backup_p (i_j, S)$ can be partitioned into a \emph{private} part $\backup_p^{\rm pr} (i_j, S)$ and a \emph{universal} part $\backup_p^{\rm un} (i_j, S)$ such that $\backup_p (i_j, S) = \backup_p^{\rm pr} (i_j, S) \cup \backup_p^{\rm un} (i_j, S)$, $\backup_p^{\rm pr} (i_j, S) \cap \backup_p^{\rm un} (i_j, S) =~\emptyset$, and
\begin{itemize}
\item[(a)]  $\backup_p^{\rm pr} (i_j, S) \subseteq P^{i_j}$, and 
\item[(b)] $\backup_p^{\rm un} (i_j, S) \cap P^{i_j} = \emptyset$.
\end{itemize}

This is simply a generalization of the $(1, 1)$-case: All jobs of $\backup_p (i_j, S)$ that intersect with $i_j$, i.e., the jobs in $\backup_p^{\rm pr} (i_j, S)$, can only be added if $i_j$ is removed from $S$ and hence they are \emph{private} for job $i_j$.
If $|\backup_p^{\rm pr} (i_j, S)| = q$, we call this a \emph{$q$-private backup for $i_j$}.
Similarly, all jobs in $\backup_p^{\rm un} (i_j, S)$ can be added to $S$ without deleting $i_j$. Hence, $\backup_p^{\rm un} (i_j, S)$ can be added to $S$ as a recourse for any $i_j \in S$. If $|\backup_p^{\rm un} (i_j, S)| = q$, we call this a \emph{$q$-universal backup}.
However, there is a fundamental difference compared to the $(1, 1)$-case: Now, the recourse after deleting $i_j$ from $S$ may consist of both, private \emph{and} universal backups.
Hence, there might exist multiple distinct $q$-universal backups for a solution $S$, depending on the deletion of the adversary.
To handle this, we show in the following decomposition lemma that the number of different $q$-universal backups is bounded.
For this, we define $U_q \coloneqq \{ \backup_p^{\rm un} (i_j, S) \mid i_j \in S \text{ and } |\backup_p^{\rm un} (i_j, S)| = q \}$, i.e., the set of distinct universal backups of $S$ of size exactly $q$. Recall that we can assume that $\backup_p (i_j, S)$ is unique and hence $\backup_p^{\rm un} (i_j, S)$ is well-defined.

\begin{figure}
\begin{center}
\begin{tikzpicture}[scale=.9]
	\interval[rwthred]{1}{.5}{2}{}{$i_1$}
	\interval[rwthred]{1}{3}{5}{}{$i_2$}
	\interval[rwthred]{1}{8}{10}{}{$i_3$}
	\interval[rwthred]{1}{12}{13.5}{}{$i_4$}
	
	\interval[gray]{2}{-2}{.5}{}{$f_0$}
	\interval[gray]{2}{2}{3}{}{$f_1$}
	\interval[gray]{2}{5}{8}{}{$f_2$}
	\interval[gray]{2}{10}{12}{}{$f_3$}
	\interval[gray]{2}{13.5}{15.5}{}{$f_4$}
	
	\interval[rwthgreen]{3}{-1.75}{-1}{}{}
	\interval[rwthgreen]{3}{2}{2.5}{}{}
	\interval[rwthgreen]{3}{5}{6}{}{}
	\interval[rwthgreen]{3}{6.5}{8}{}{}
	\interval[rwthgreen]{3}{10}{11}{}{}
	\interval[rwthgreen]{3}{14}{15.5}{}{}
	\interval[rwthgreen]{4}{-1.5}{-.5}{}{}
	\interval[rwthgreen]{4}{2.25}{3}{}{}
	\interval[rwthgreen]{4}{5.5}{7.25}{}{}
	\interval[rwthgreen]{4}{10.5}{12}{}{}			
	
\end{tikzpicture}        
\end{center}
\caption{The set of red intervals consists of a solution, and the set of gray intervals are the corresponding free intervals which together form the set {\rm Free}. Additionally, the green intervals form the set $I_{\rm Free} \subseteq I$.}
\label{fig:(1,p)RIS}
\end{figure}

\begin{lemma}
\label{lem:1-p-interval:decomposition}
Let $I$ be an instance of $(1,p)$-\robustintervalselection and let $S$ be an optimum first-stage solution of $I$. Then the following holds.
\begin{enumerate}
\item $\backup_p (i_j, S)$ can be decomposed into a private part $\backup_p^{\rm pr} (i_j, S)$ and a universal part $\backup_p^{\rm un} (i_j, S)$ as above.
\item There are at most $2q+1$ many $q$-universal backups, i.e., $|U_q| \leq 2q+1$.
\end{enumerate}
\end{lemma}

\begin{proof}
    The first part is trivial due to our discussion above. Hence, it remains to prove the second part.
    Let $S = \{ i_1, \ldots, i_m\}$ be the first-stage solution ordered non-decreasingly by starting times of the jobs, where $i_j = [a_j, b_j)$.
    Let $a_0$ be the smallest starting time of any job in $I$ and let $b_n$ be the largest finishing time of any job in $I$.
    Define the \emph{free} intervals $f_0 = [a_0, a_1)$, $f_1 = [b_1, a_2)$, $f_2 = [b_2, a_3)$, \ldots, $f_m=[b_m, b_n)$ to be the intervals that are not occupied by $S$ and let ${\rm Free} = \{f_0, f_1, \ldots, f_m\}$.
    Let $I_{{\rm Free}} \subseteq I$ be the set of jobs that lie completely within some $f_i \in {\rm Free}$.
    See Figure~\ref{fig:(1,p)RIS} for reference.
    By definition of universal backups, any universal backup $i_k \in I \setminus S$ cannot intersect a job of $S$ and hence must lie in some $f_i \in {\rm Free}$. Therefore, $i_k \in I_{{\rm Free}}$.
    Note, however, that the private backup $\backup_p^{\rm pr} (i_j, S)$ of some job $i_j \in S$ can intersect $f_{j-1}$ and $f_j$ but no other $f_i$, $i \notin \{j-1, j\}$. 
    Otherwise, if $\backup_p^{\rm pr} (i_j, S)$ intersects some $f_i$, $i \notin \{j-1, j\}$, $\backup_p^{\rm pr} (i_j, S)$ is not a feasible backup,  since $\backup_p^{\rm pr} (i_j, S)$ is a subset of $P^{i_j}$ and hence it must intersect $i_{j-1}$ or $i_{j+1}$. 
    
    Let $R_f^q \subseteq I_{{\rm Free}}$ be the maximum weight set of at most $q$ jobs from $I_{{\rm Free}}$.
    Note that $R_f^q$ can intersects at most $q$ intervals of ${\rm Free}$ and let these intervals be ${\rm Free}' \subseteq {\rm Free}$.
    By the properties of a private backups, 
    for any $j$ such that $f_{j-1}, f_j \notin {\rm Free}'$, the job $i_j \in S$ can have $R_f^q$ as a $q$-universal backup. 
    Since $R_f^q \subseteq I_{{\rm Free}}$ is the maximum weight subset 
    of jobs of $I_{{\rm Free}}$ that selects at most $q$ jobs, we have $\backup_p^{\rm un} (i_j, S) = R_f^q$ (recall that we assume the backup to be unique).
    Furthermore, since $|R_f^q| \leq q$ and hence $|{\rm Free}'|\leq q$, there are at most $2q$ jobs $i_j \in S$ such that $\backup_p^{\rm pr} (i_j, S)$ intersects with a job in $R_f^q$.
    Each such job might have its unique $q$-universal backup $\backup_p^{\rm un} (i_j, S)$.
    Hence, in total there are at most $2q+1$ many $q$-universal backups, which proves the Lemma.
\end{proof}

Similar to the $(1, 1)$-case, we also have \emph{extra} backups:
For the case where $\lambda<0$ and a job $i_j$ of a $p$-universal backup is interdicted, we need an extra backup.
However, note that there is at most one $p$-universal backup, since the recourse is bounded by $p$ and in that case the private backup for any job is empty. 
Therefore, the $p$-universal backup of any job $i_j \in S$ is $R^p$, where $R^p$ is the maximum weight set of jobs that does not intersect $S$ and contains at most $p$ jobs. 
Since any job $i_k$ of $R^p$ can be removed, for each such case we need at most one extra backup of size $p$ that does not contain any job from $S \cup \{i_k\}$, and hence the number of extra backups is bounded by $p$.
This leads to the following lemma.

\begin{lemma}
\label{lem:1-p-interval:bound-universal-and-extra-backups}
Let $I$ be an instance of $(1,p)$-\robustintervalselection and let $S$ be an optimum first-stage solution of $I$. 
Then the total number of universal and extra backups is bounded by $p \cdot (p+3)$.
\end{lemma}

\begin{proof}
    Due to Lemma~\ref{lem:1-p-interval:decomposition}, we know that the total number of $q$-universal backups is bounded by $2q+1$. 
    Hence, the total number of universal backups is bounded by
    \[
        \sum_{q=1}^p (2q+1) = p + 2 \cdot \sum_{q=1}^p q = p + p \cdot (p+1) = p \cdot (p+2) \ .
    \]
    By our previous discussion, the number of extra backups is bounded by $p$ and hence the total number of universal and extra backups is bounded by 
    \[
        p \cdot (p+2) + p = p \cdot (p+3) \ ,
    \]
    which proves the lemma.
\end{proof}

Lemma~\ref{lem:1-p-interval:bound-universal-and-extra-backups} allows us to guess the set of universal and extra backups of an optimum solution: The size of each set is bounded by $p$ and there are at most $p \cdot (p+3)$ many distinct sets. Hence, the total number of guesses we need to perform is $\mathcal{O}((n^p)^{p \cdot (p+3)})$, which is polynomial if $p$ is constant.

After guessing the value of $\lambda$ and guessing the set of universal and extra backups of an optimum solution, we construct an instance of \issh, similar to the $(1, 1)$-case. Here, it remains to construct, for each $i_j \in I$, the appropriate private backups.
Recall that a $q$-private backup of $i_j \in I$ must intersect $P^{i_j}\coloneqq \{i_k \in I \mid [a_k, b_k) \cap [a_j, b_j) \neq \emptyset\}$. Each $P^{i_j}$, $i_j\in I$, can be calculated in $\mathcal{O}(n)$ time.
To obtain all possibilities for a $q$-private backup, for each $q \in \{1, 2, \ldots, p\}$ we enumerate over all starting points $\ell$ and ending points $r$ of that private backup such that $\ell \leq r$.
For each $q$ and each such pair of endpoints $(\ell, r)$, we compute in polynomial time the maximum weighted set of pairwise disjoint intervals in the set $\{i_k\in P^{i_j} \mid l\leq a_k < b_k\leq r\}$ that contains at most $q$ intervals. 
The total number of $(\ell,r)$ pairs possible is $\binom{2n}{2}$.
All in all, we can calculate $q$-private backups for all $q\in[p]$ of all jobs in polynomial time. 

Finally, we construct the instance of \issh by computing all $q$-private backups for each $i_j \in I$ and define the corresponding jobs with red and blue intervals. Afterwards, we prune the resulting jobs from the instance of \issh according to the guessed value of $\lambda$, where we consider for each job its $q$-private backup and the best feasible $(p-q)$-universal backup (from the set of guessed universal backups).
Then the proof of Theorem~\ref{thm:rros:1,p:RIS} is similar to the proof of Theorem~\ref{thm:interval-selection-poly}.
	\section{Robust Bipartite Matching}\label{sec:robust-bipartite-matching}
{\sc Maximum Matching} is a classical combinatorial optimization problem in which we are given an undirected graph $G$, and the task is to find a matching in $G$ of maximum cardinality.
It is well known that one can find a maximum-cardinality matching in polynomial time even in general graphs.
In stark contrast to this, we show that the robust counterpart of {\sc Maximum Matching} becomes intractable even when we restrict ourselves to bipartite graphs.
In this section, we consider the problem \robustmatching\ (\robustmatchingshort), which is the problem \eqref{eq: (1,1)recoverablerobustproblem} where $\mathcal{F}$ consists of all matchings in an underlying bipartite graph.
We study the decision version of this problem.
Note that the problem of finding a maximum-cardinality matching in a bipartite graph is a classical example for {\sc Matroid Intersection}, i.e., the problem to find a maximum-cardinality set which is independent in two matroids on the same ground set $E$.
Hence, the following result illustrates a drastic increase in the complexity when considering the {\em intersection of two matroids} instead of a {\em single matroid}.
Below we give a definition of the {\sc Unweighted} \robustmatching problem.

\begin{problemenv}
	\problemtitle{{\sc Unweighted} \robustmatching (\robustmatchingshort)}
	\probleminput{An undirected bipartite graph $G = (V, E)$ and a positive integer $k$.}
	\problemquestion{Is there a matching $M$ in $G$ such that for each $\f \in E$ there exists an edge $\e\in E-\f$ for which $M-\f+\e$ is a matching of cardinality at least $k$?}
\end{problemenv}

Before we proceed, we make a general statement on Robust maximum-cardinality problems.
Consider the unweighted version of problem~\eqref{eq: (1,1)recoverablerobustproblem}, where the task is to find a set $S \in \mathcal{F}$ which maximizes $|S - \f + \e|$ after the interdiction of $\f \in E$ and the possible recourse of element $\e\in E - \f$.
Note that $|S-\f+\e| \geq |S|-1$. 
Thus, it suffices to restrict our search of $S$ to the maximum cardinality sets in $\mathcal{F}$. 
Moreover, it suffices to consider $\f\in S$.
We call a solution $S \in \mathcal{F}$ of the nominal problem \emph{repairable}, if for each $\f \in S$ there exists an element $\e \in E \setminus S$ such that $S-\f+\e \in \mathcal{F}$.

\begin{observation}\label{obs:unit-weights}
	For the robust counterpart \eqref{eq: (1,1)recoverablerobustproblem} of a combinatorial problem \problempi with unit weights, there is a maximum-cardinality solution $S$ to \problempi that is optimal for~\eqref{eq: (1,1)recoverablerobustproblem}.
	Furthermore, if there is a maximum-cardinality solution $S$ of~ \problempi that is repairable, then $S$ is optimal.
\end{observation}

\noindent We are now ready to state and prove the main result of this section.

\begin{restatable}[]{theorem}{matching}
	\label{thm:hardness:matching}
	{\sc Unweighted} \robustmatching problem is \NP-complete.
\end{restatable}
\begin{proof}
	Clearly, the {\sc Unweighted Robust  Bipartite Matching}  problem is in \NP, because for any given matching $M$, one can enumerate for each possible deletion of edge $\f \in E$ the recourse $\e \in E-\f$, and verify whether \mbox{$|M-f+e| \geq k$.}
	
	If we set $k$ to be the size of a maximum-cardinality matching, then by Observation~\ref{obs:unit-weights}, we need to essentially find a repairable matching in the graph.
	Thus, it suffices to show that deciding whether there exists a \repairable matching of maximum cardinality is \NP-hard.
	We refer to this problem as \repairablematching (\repairablematchingshort) and give a reduction from \tsat: given a set of Boolean variables $X=\{x_1, x_2, \dots, x_n\}$ and a formula $\phi$ in conjunctive normal form, where each clause has $3$ variables, does there exist a truth assignment that satisfies~$\phi$?
	
	Let $I$ be an instance of \tsat\ with a set $X$ of Boolean variables and a set $Y=\{C_1, C_2, \dots C_m\}$ of clauses in which $|C_i|=3$. 
	We construct an instance \irm of the problem \repairablematchingshort with input graph $G=(V,E)$ as follows (see also Figure~\ref{fig:a:reduction}): 
	\begin{itemize}
		\item For each variable $x_i\in X$ in $I$, we add two nodes $a_i$ and $\bar{a}_i$ in $V(G)$ corresponding to a positive and negative literal of $x_i$, respectively. 
		We set $A \coloneqq \{a_1,\dots, a_n, \bar{a}_1, \dots, \bar{a}_n\}$.
		Moreover, we add $b_i$ in $V(G)$ for each $x_i\in X$ and set $B \coloneqq \{b_1, \dots, b_n \}$. 
		Additionally, for each clause $C_j\in Y$, we add two nodes $c_j, z_j$ in $V(G)$ and set $C \coloneqq \{c_1, \dots, c_m \}$ and $Z \coloneqq \{z_1, \dots, z_m \}$.
		Finally, we have $V(G) = A \cup B \cup C \cup Z$.
		\item Let the edge set $E(G)$ be the union of the following set of edges. 
		\begin{itemize}
			\item[(a)] For all $i\in [n]$, add $a_ib_i, \bar{a}_ib_i$ to $E(G)$.
			\item[(b)] For all $j\in [m]$, add $c_jz_j$ to $E(G)$.
			\item[(c)] If $x_i\in C_j$, then add $a_ic_j$ to $E(G)$ and if $\bar{x}_i\in C_j$, then add $\bar{a}_ic_j$ to $E(G)$.
		\end{itemize}
	\end{itemize}
	Note that $G$ is bipartite as $V(G)$ can be partitioned into two stable sets, $A \cup Z$ and $B \cup C$.
	\begin{figure}
		\begin{center}
			\begin{tikzpicture}[scale=0.58, >=latex, rotate=90]
				\draw[black!40, thick, dashed, rounded corners] (-1,2) -- (11, 2) --(11,-18)-- (-1, -18) -- cycle;
				\node (phi) at (10,-8) {\small
					$\phi = (x_1 \vee x_2 \vee \bar{x}_3) \wedge (\bar{x}_1 \vee  \bar{x}_2 \vee x_n) \wedge \dots \wedge (x_1 \vee x_3 \vee \bar{x}_n)$};
				
				\node (to) at (8.9,-8) {$\downarrow$};
				
				\foreach \x in {-2,-5,-14} {\draw[black] (6,\x) -- (7.5,\x);}
				\foreach \x in {-2,-5,-8,-14} {\draw[black] (0,\x) -- (1.5,\x-1+0.3);}
				\foreach \x in {-2,-5,-8,-14} {\draw[black] (0,\x) -- (1.5,\x+1-0.3);}
				
				\foreach \x in {-2,-5,-14} {\draw[black] (6,\x) -- (7.5,\x);}
				
				\foreach \x in {-5,-8} {\draw[black] (0,\x) -- (1.5,\x-1+0.3);}
				\foreach \x in {-2, -14} {\draw[black] (0,\x) -- (1.5,\x-1+0.3);}
				\foreach \x in {-2,-5,-8,-14} {\draw[black] (0,\x) -- (1.5,\x+1-0.3);}
				\foreach \x in {-5,-8} {\draw[black] (0,\x) -- (1.5,\x+1-0.3);}
				
				\draw[black] (1.5,-1-0.3) -- (6,-2);
				\draw[black] (1.5,-4-0.3) -- (6,-2);
				\draw[black] (1.5,-9+0.3) -- (6,-2);
				\draw[black] (1.5,-3+0.3) -- (6,-5);
				\draw[black] (1.5,-6+0.3) -- (6,-5);
				\draw[black] (1.5,-13-0.3) -- (6,-5);
				\draw[black] (1.5,-1-0.3) -- (6,-14);
				\draw[black] (1.5,-7-0.3) -- (6,-14);
				\draw[black] (1.5,-15+0.3) -- (6,-14);
				
				\foreach \x in {-2,-5,-8,-14} {\node[bluevertex] (\x) at (0,\x*1) {};}
				\foreach \x in {-2,-5,-14} {\node[bluevertex] (\x) at (6,\x*1) {};}
				\foreach \x in {-2,-5,-14}  {\node[redvertex] (\x) at (7.5,\x*1) {};}
				\foreach \x in {-1,-4,-7,-13} {\node[redvertex] (\x) at (1.5,\x*1-0.4) {};}
				\foreach \x in {-3,-6,-9,-15} {\node[redvertex] (\x) at (1.5,\x*1+0.4) {};}
				
				\node[color=black!40] (ldots) at (1.5,-11) {$\hdots$};
				\node[color=black!40] (ldots) at (7,-9.5) {$\hdots$};
				
				\node (y1) at (-0.1,-2.5) {\small$b_1$};
				\node (y2) at (-0.1,-5.5) {\small$b_2$};
				\node (y3) at (-0.1,-8.5) {\small$b_3$};
				\node (yn) at (-0.1,-14.5) {\small$b_n$};
				
				\node (c1) at (6.35,-2.35) {\small$c_1$};
				\node (c2) at (6.35,-5.35) {\small$c_2$};
				\node (cm) at (6.35,-14.45) {\small$c_m$};
				
				\node (z1) at (7.85,-2.3) {\small$z_1$};
				\node (z2) at (7.85,-5.3) {\small$z_2$};
				\node (zm) at (7.85,-14.4) {\small$z_m$};
				
				\node (x1) at (1.33,-0.95) {\small$a_1$};
				\node (x2) at (1.33,-3.95) {\small$a_2$};
				\node (x3) at (1.33,-6.95) {\small$a_3$};
				\node (xn) at (1.33,-12.95) {\small$a_n$};
				
				\node (x1') at (1.35,-3.15) {\small$\bar{a}_1$};
				\node (x2') at (1.35,-6.15) {\small$\bar{a}_2$};
				\node (x3') at (1.35,-9.15) {\small$\bar{a}_3$};
				\node (xn') at (1.35,-15.15) {\small$\bar{a}_n$};
			\end{tikzpicture}
		\end{center}
		\caption{\NP-hardness reduction from \tsat to {\sc Unweighted} \robustmatching.}
		\label{fig:a:reduction}
	\end{figure}
	Given a satisfying truth assignment $\gamma: X\rightarrow\{0,1\}$ of $I$, we show that the matching  $M=\{c_j z_j \mid j \in [m]\}\cup \{a_i b_i\mid \gamma(x_i)=0\}\cup \{\bar{a}_ib_i\mid \gamma(x_i)=1\}$ is a solution for \irm.
	Notice that $|M|=m+n$ as $M$ leaves exactly $n$ unsaturated vertices in $V(G)$. 
	The set of these $n$ unsaturated vertices constitutes of one element from each $\{a_i,\bar{a}_i\}$ where $i\in [n]$.
	These are precisely the literals with truth value 1 assigned to them by $\gamma$. 
	All edges $\e\in M$ contain either some $b_i$, $i \in [n]$, or some $c_j$, $j \in [m]$, as an endpoint. 
	Note that $M$ is a matching in $G$ of maximum cardinality.
	
	To prove that $M$ is a solution of \irm,  we show that for each edge $\f\in M$ there is an edge~$\e\in E(G)\setminus M$ such that $M-\f+\e$ is a matching. 
	In other words, for each edge $\f$, there is an edge $\e$ connecting an unsaturated vertex in $V(G)$ to an endpoint of $\f$. 
	If $\f$ is incident to $b_i$, we take as $\e$ the other edge incident to $b_i$. 
	If $\f$ is incident to $c_j$, then there is an unsaturated vertex $a_i$ such that $a_ic_j\in E(G)\setminus M$, or an unsaturated vertex $\bar{a}_i$ such that $\bar{a}_ic_j \in E(G)\setminus M$.
	This is because~$M$ is derived from the satisfying truth assignment $\gamma$ of $I$, and clause $C_j$ contains some literal $x_i$ or $\bar{x}_i$ with truth assignment~$1$. 
	
	Next, we show that if \irm is a yes-instance of \repairablematchingshort, then $I$ also is a yes-instance of \tsat. 
	If \irm is a yes-instance of \repairablematchingshort, then there is a matching $M'$ of size $m+n$ in which each $\f \in M'$ has an edge $\e\in E(G)\setminus M'$ such that $M'-\f+\e$ is a matching. 
	
	First, we claim that we can assume without loss of generality that $M'$ contains all edges in $\{ c_j z_j \mid j \in [m] \}$.
	Assume this is not true. 
	Then, since $M'$ has to have size $m+n$, there must exist some edge $a_i c_j \in M'$ or some edge $\bar{a}_i c_j \in M'$. 
	Without loss of generality, assume we are in the former case.
	Then we claim that $M'' = M' + c_j z_j - a_i c_j$ is also a \repairable matching of cardinality $m+n$.
	Clearly, $M''$ is a matching of cardinality $m+n$.
	Hence, it remains to show that for each $\f \in M''$ there exists some $\e \in E \setminus M''$ such that $M'' - \f + \e$ is also a matching.
	First, consider some $\f \in M'' \setminus \{ c_j z_j \}$. 
	In that case, set $\e$ to be some edge such that $M' - \f + \e$ is a matching (which exists since $M'$ is \repairable).
	Observe that $\e$ cannot be incident to $c_j$ (since $a_i c_j \in M'$) and also not incident to $z_j$ (since $z_j$ has degree 1).
	Hence, $M'' - \f + \e$ is a matching.
	Second, consider the case $\f =  c_j z_j$. 
	Now choose $\e = a_i c_j$. 
	In this case, clearly, $M'' - \f + \e$ is a matching since $M'' - \f + \e  = M'$.
	By applying this procedure iteratively to all edges $a_i c_j \in M'$, we obtain the desired result.
	
	Observe that all edges in $M'$ are either of the form $c_j z_j$ for some $j \in [m]$, or of the form $a_i b_i$ or $\bar{a}_i b_i$ for some $i \in [n]$.
	Hence, since $M'$ has cardinality exactly $m+n$, either $a_i$ or $\bar{a}_i$, $i \in [n]$ is saturated by~$M'$.
	Let $U\subseteq A$ be the set of vertices in $A$ that are not saturated by~$M'$. 
	Let $\gamma: X\rightarrow \{0,1\}$ where $\gamma(x_i)=1$ if $a_i\in U$, and $\gamma(x_i)=0$ if $\bar{a}_i\in U$. 
	We claim that this assignment satisfies the formula~$\phi$. 
	
	For the sake of contradiction, assume that $\gamma$ does not satisfy $\phi$. 
	Then there is a clause $C_j$ for which all three literals in $C_j$ evaluate to 0. 
	Now consider the corresponding vertices in the graph~$G$; the literals evaluate to 0 because the corresponding vertices are matched to $b_i$ in matching $M'$. 
	This implies there is no $\e\in E(G)\setminus M'$ which replaces the edge $c_jz_j \in M'$. 
	However, this contradicts the fact that $M'$ is a solution to \irm. 
\end{proof}

	\section{Robust Stable Set in Bipartite Graphs}\label{sec:robust-independent-set}
This section is devoted to the \robuststableset (\robuststablesetshort) problem, which is the problem \eqref{eq: (1,1)recoverablerobustproblem} where the feasibility set $\mathcal{F}$ is the set of all stable sets, commonly also known as independent sets, of the input graph.
In the following we consider bipartite graphs and show
that the {\sc Unweighted} \robuststableset problem admits an efficient algorithm in bipartite graphs.
This result extends to 
\emph{König-Egerváry} graphs which generalize bipartite graphs; 
they are graphs in which the size of a minimum vertex cover equals the size of a maximum matching \cite{deming1979independence}.
Finally, we show that the \emph{weighted} problem is \NP-hard even in bipartite graphs.

\begin{theorem}\label{maxcardindbip}
	{\sc Unweighted} \robuststableset is polynomial-time solvable in bipartite graphs. 
\end{theorem}

By Observation~\ref{obs:unit-weights}, it suffices to show that we can find in polynomial time a \repairable maximum-cardinality stable set in bipartite graphs if it exists, or declare that none exist. 
Hence, Theorem~\ref{maxcardindbip} follows directly from the following lemma.

\begin{lemma}
	\label{lem:repairable-stable-set}
	A bipartite graph $G$ contains a \repairable maximum-cardinality stable set if and only if the underlying graph admits a perfect matching such that each edge in the matching has a degree-one endpoint.
\end{lemma}

Note that in order to check whether a graph admits a perfect matching such that each edge in the matching has a degree-one endpoint, it suffices to compute a maximum-cardinality matching $M$, and check whether $M$ is perfect, and whether each edge of $M$ has at least one endpoint of degree $1$.
\begin{proof}
	Let $G$ admit a maximum stable set  $S\subseteq V(G)$ which is \repairable.
	Then by definition, for all $v\in S$, there is a vertex $w\in V\setminus S$ such that $S-v+w$ is also a stable set. 
	Note that, if  $w\notin N(v)$,  where $N(v) \coloneqq \{ w \in V \mid \{v, w \} \in E \}$ is the set of neighbors of $v$, then $S+w$ is a stable set in contradiction to the fact that $S$ is a stable set of maximum cardinality in $G$.
	Hence for all $v\in S $, there is a vertex $w\in N(v)$ such that $S-v+w$ is a stable set. 
	We call $w$ the \textit{backup} vertex of $v$.
	Note that $N(w) \cap S = \{ v \}$.
	This implies that in a \repairable stable set $S$, each vertex $v\in S$ has a distinct backup vertex $w$.
	Hence, there should be at least $|S|$ many backup vertices in the graph (one for each $v \in S$), and since the backup vertices do not belong to the stable set $S$, we have $|S|\leq \frac{n}{2}$. 
	As the size of a maximum stable set in a bipartite graph is at least $\frac{n}{2}$, we have $|S|=\frac{n}{2}$.  
	Therefore, $G$ contains a perfect matching $M = \{vw\in E\mid v\in S, w \text{ is a backup of } v\}$. 
	Now consider vertex $v\in S$ and its backup vertex $w\in N(v)$; $w$ is the only neighbor of $v$, otherwise $S$ is not a stable set or some vertex $u\in V\setminus S$ is not a backup vertex. 
	Thus each $v \in S$ is degree-one.
	
	Now let us assume that $G$ has the desired properties. We show that there is a \repairable\  maximum-cardinality stable set in $G$.
	If the graph contains a perfect matching and every matching edge is incident to a degree one vertex,	then a collection of these degree-one vertices, $S \subseteq V(G)$, one from each edge in the perfect matching, is a solution of the instance of the problem. 
	Observe that for each vertex $v\in S$, the unique neighbor $w$ of $v$ serves as a backup vertex.
\end{proof}

Notice that the crucial properties used to prove Lemma~\ref{lem:repairable-stable-set} and Theorem~\ref{maxcardindbip} even hold for König-Egerváry graphs and thus we get the following. 
\begin{corollary}
	\label{cor:konig-egervary}
	{\sc Unweighted} \robuststableset  is polynomial-time solvable in König-Egerváry graphs. 
\end{corollary}

This is interesting in the following sense:
Theorem~\ref{thm:hardness:matching} implies a hardness result for 
{\sc Unweighted} \robustmatching.
Recall that the line graph of a graph $G$ contains the vertex set corresponding to $E(G)$, and there is an edge between two vertices if and only if the corresponding edges share an endpoint in~$G$.
In general, a maximum stable set in the line graph of a graph $G$ corresponds to a maximum matching in $G$, and hence can be computed in polynomial time.
However, while Theorem~\ref{maxcardindbip} gives a polynomial-time algorithm for {\sc Unweighted} \robuststableset in bipartite graphs, this one-to-one correspondence and Theorem~\ref{thm:hardness:matching} imply that 
{\sc Unweighted} \robuststableset 
is \NP-hard in line graphs of bipartite graph.

\begin{corollary}
	{\sc Unweighted} \robuststableset is \NP-hard in line graphs of bipartite graphs.
\end{corollary}

Next, we show that the complexity of \robuststableset in bipartite graphs changes with weights.
We now consider the decision version of the \robuststableset problem.

\begin{problemenv}
	\problemtitle{\robuststableset (\robuststablesetshort)}
	\probleminput{An undirected graph $G$, a weight function $w \colon V \to \mathbb{R}_+$, and a positive integer $k$.}
	\problemquestion{Is there a stable set $S$ in $G$ such that for each $v \in V$ there exists a node $v' \in V-v$ which gives another stable set  $S-v+v'$ such that $w(S-v+v') \geq k$?}
\end{problemenv}

\begin{restatable}[]{theorem}{weightedstableset}
	\label{thm:hardness:weighted-stable-set}
	\robuststableset is \NP-complete in bipartite graphs.
\end{restatable}
\begin{proof}
	We prove that the above mentioned decision variant of the problem \weightedind\ is \NP-hard in bipartite graphs.
	We give a reduction from \tsat. Given an instance $I$ of \tsat\ with the set of variables $X=\{x_1, x_2, \dots, x_n\}$ and a satisfiability formula $\phi$ with clauses $Y=\{C_1, C_2, \dots C_m\}$, we construct an instance $I' = (G, w, k)$ of the \weightedind\ problem as follows:
	\begin{itemize}
		\item For each variable $x_i\in X$, we take vertices $a_i, \bar{a}_i, b_i, b_i^1, b_i^2$ in $V(G)$. 
		For each clause $C_j$ in $\phi$, we add a vertex $c_j$ in $V(G)$. 
		Moreover, we add 3 representative vertices, $c_j^1,c_j^2,c_j^3$, corresponding to each literal in $C_j$.
		\item The edge set $E(G)$ consists of the following set of edges:
		\begin{itemize}
			\item[(a)] For all $i\in[n]$, $b_ib_i^1, b_ib_i^2, b_i^1a_i, b_i^2\bar{a}_i \in E(G)$.
			\item[(b)] For all $j\in[m]$, $c_jc_j^1, c_jc_j^2, c_jc_j^3 \in E(G)$. 
			\item[(c)] Connect the three vertices $c_j^1, c_j^2, c_j^3$ to the corresponding literals in the clause, e.g.\ if $C_1=(v_1\vee\bar{v}_2\vee v_4)$, then add edges $c_1^1a_1, c_1^2\bar{a}_2, c_1^3a_4\in E(G)$. 
		\end{itemize}
		\item For the weight function $w: V(G)\rightarrow \mathbb{R}_{\geq0}$, consider $1 \ll r \ll s$, $r,s\in \mathbb{R}_{\geq0}$ and set 
		\begin{align*}
			w(v) & = 
			\begin{cases}
				s &\quad\text{if } v \in \{c_j\mid j\in[m]\} \cup \{b_i\mid i\in[n]\},\\
				r &\quad\text{if } v \in \{c_j^1, c_j^2, c_j^3\mid j\in[m]\}\cup \{b_i^1, b_i^2\mid i\in[n]\},\\
				1 &\quad\text{if } v \in \{a_i, \bar{a}_i\mid i\in[n]\} \ .
			\end{cases} 
		\end{align*}
		\item Set $k=(m+n-1)s+r+n$.
	\end{itemize}
	Note that the constructed graph is bipartite. 
	
	First, let us assume that $\phi$ is a yes-instance and let $\gamma: X \rightarrow \{0,1\}$ be the truth assignment. 
	We claim that $S=\{c_j\mid j\in[m]\} \cup \{b_i\mid i\in[n]\}\cup \{a_i\mid \gamma(x_i)=0\}\cup \{\bar{a}_i\mid \gamma(x_i)=1\}$ is a solution to~$I'$ of weight at least $k$.
	First, observe that $S$ is a stable set and that $w(S) = (m+n)s+n = k + s - r$. 
	Moreover, note that $a_i\notin S$ if and only if $\gamma(x_i)=1$. 
	For all $j\in[m]$, there exists a vertex $c_j'\in\{c_j^1,c_j^2, c_j^3\}$ such that $c_j'$ corresponds to a literal $a$ for which $\gamma(a)=1$, 
	and hence $a\notin S$.
	Thus, $c_j'$ is only connected to $c_j$ in $S$. 
	Therefore, this vertex $c_j'$ can be a backup vertex of $c_j$, meaning that $S + c_j' - c_j$ is a stable set of value $k$. 
	Also, with a similar argument, there is a backup vertex~$b_i^1$ or $b_i^2$ for every $b_i\in S$. 
	Thus, for each $v\in S$ of weight $w(v)=s$, there is a backup vertex $v'$ of weight $w(v')=r$. 
	Consequently, for each $v\in S$ there exists some $v' \in V \setminus S$ such that
	$w(S) + w(v') - w(v)\geq (m+n-1)s+r+n = k$,
	as claimed.
	
	Next, let us assume that there is a stable set $S$ of $G$ such that for each $v \in S$ there is some $v' \in V \setminus S$ with 
	$w(S) + w(v') - w(v) \geq (m+n-1)s+r+n = k$.
	We construct a truth assignment~$\gamma$ for $I$.
	Observe that $S$ must contain all $m+n$ vertices of weight $s$, and none of the weight $r$ vertices since they are all neighbors of weight $s$ vertices.
	Furthermore, by the property of $S$, for each $c_j$ there is a vertex $c_j'\in \{c_j^1, c_j^2, c_j^3\}$, such that the corresponding literal of $c_j'$, say $a$, is not in $S$, i.e., $S + c_j' - c_j$ is a stable set (i). 
	Since for all $v \in S$ there exists some $v' \in V \setminus S$ such that $w(S) - w(v) + w(v')\geq (m+n-1)s+r+n$, the set $S$ must contain $n$ vertices from set $A=\{a_1, \dots, a_n, \bar{a}_1,\dots, \bar{a}_n\}$. 
	Furthermore, since for each $b_i$ there is a vertex $v'$ of value $r$ such that $S + v' - b_i$ is a stable set, not both $a_i$ and $\bar{a}_i$ are in $S$. 
	Hence, exactly one of $a_i,\bar{a}_i$ belongs to $S$ (ii). 
	We take the set of vertices in $A$ which are not part of $S$ and define the assignment $\gamma(x_i)=1$ if and only if $a_i\notin S$. 
	We claim that $\gamma$ is a truth assignment.
	Clearly, by (ii), $\gamma$ is an assignment;
	and by (i), $\gamma$ satisfies exactly the definitions of a truth assignment. 
\end{proof}

	\section{Conclusion and Future Research} 
		In this paper we have introduced the model {\em \generalproblem} and investigated it for various combinatorial optimization problems, settling their computational complexity.
		While our focus was on polynomial-time solvable problems with a downward-closed feasibility set, it would be interesting to extend the study to a broader class of problems including \NP-hard problems. 
		Problems whose feasibility set is not downward-closed pose the additional challenge that the deletion of some elements might destroy feasibility which needs to be repaired by the recourse action. 
		
		Furthermore, we often have restricted our considerations to problem settings, where the adversary is only allowed to delete a single element (i.e., $k=\ell=1$), which seems most fundamental. 
		It would be interesting to see how the computation complexity unravels for the general setting with arbitrary valued $k$ and $\ell$.
		For $(k,\ell)$-\robustmatroidbasis with $k\le \ell$, we showed that the nominal optimal solution is also optimal for the robust counterpart, and gave an example illustrating that this property does not necessarily hold for $k>\ell$, or non-matroid structures.
		It still remains an open question whether or not $(k,\ell)$-\robustmatroidbasisshort is polynomial-time solvable  if $k> \ell$. 
		Similarly, even though we have efficient algorithms for $(1,p)$-\robustintervalselection for a constant value $p$, the complexity status of $(k,\ell)$-\robustintervalselection and also the $(k,\ell)$-\robuststableset remain open.
		
		Finally, the bounded-regret problem was crucial for solving \robustintervalselection. It remains open whether the bounded-regret version of other problems can be solved in polynomial time. Particularly interesting is \robustmatroidbasis which we can solve efficiently for $k\leq \ell$, but the complexity of the bounded-regret version is open.
		
    \subsection*{Acknowledgements}
    	We thank Ruiyang Zhang for discussions on Robust Interval Scheduling during his master’s thesis and for observing Example \ref{example:rroc:matroid:counterexample}.
    	
    \bibliographystyle{plain}
    \bibliography{references}		
\end{document}